\documentclass{amsart}

\usepackage{hyperref}
\usepackage{amsmath}
\usepackage{amssymb}
\usepackage{amscd}
\usepackage{graphicx}
\usepackage{mathdots}
\usepackage{gensymb}
\usepackage{epstopdf}
 \usepackage{todonotes}
 \usepackage{youngtab}
 \usepackage{wrapfig}

\usepackage{epsfig}       
\usepackage{epic,eepic}        

\newtheorem{thm}{Theorem}[section]

\newtheorem{lem}[thm]{Lemma}

\newtheorem{conj}[thm]{Conjecture}

\theoremstyle{definition}

\theoremstyle{remark}
\newtheorem{remark}[thm]{Remark}

\numberwithin{equation}{section}

\newcommand{\group}[1]{\mathrm{#1}}

\newcommand{\mon}{\vec{H}}

\newcommand{\Tr}{\operatorname{Tr}}

\newcommand{\End}{\operatorname{End}}

\newcommand{\Cat}{\operatorname{Cat}}

\newcommand{\D}{\mathbf{D}}
\renewcommand{\O}{\mathcal{O}}

\newcommand{\N}{\mathbb{N}}
\renewcommand{\P}{\mathbb{P}}
\newcommand{\Z}{\mathbb{Z}}

\newcommand{\LIS}{\mathrm{LIS}}

\DeclareRobustCommand{\stirling}{\genfrac\{\}{0pt}{}}

\newcommand{\C}{\mathbb{C}}
\newcommand{\R}{\mathbb{R}}

\begin{document}


\title[Topological Expansion of Oscillatory BGW and HCIZ Integrals]{Topological Expansion of Oscillatory BGW and HCIZ Integrals at Strong Coupling}

 
\author{Jonathan Novak}
\address{Department of Mathematics, University of California, San Diego, USA}
\email{jinovak@ucsd.edu}



 \maketitle
 
 
 
 \tableofcontents
 \pagebreak



\section{Introduction}
\label{sec:Intro}

\subsection{Overview}
The purpose of this paper is to prove a longstanding conjecture on 
the asymptotic behavior of two particularly significant
matrix integrals: the Bars-Green/Br\'ezin-Gross-Witten/Wadia integral,

	\begin{equation}
		\label{eqn:BGW}
		I_N^{(1)} = \int_{\mathrm{U}(N)} e^{\sqrt{z}N \mathrm{Tr}(AU + BU^{-1})} \mathrm{d}U,
	\end{equation}

\noindent
and the Harish-Chandra/Itzykson-Zuber integral,

	\begin{equation}
		\label{eqn:HCIZ} 
		I_N^{(2)} = \int_{\mathrm{U}(N)} e^{zN \mathrm{Tr} AUBU^{-1}} \mathrm{d}U.
	\end{equation}
	
\noindent
These are integrals over unitary matrices against unit mass Haar measure which depend on a complex parameter $z$ 
and a pair of complex matrices $A$ and $B$. The BGW and HCIZ integrals arise in many contexts throughout mathematics
and physics, ranging from representation theory and random matrix theory to quantum field theory and statistical mechanics,
and the following conjecture on their $N \to \infty$ asymptotics is widely known; see \cite{Collins,CGM,GGN3,Guionnet:PS,GN} and 
further references below.

	\begin{conj}[Topological Expansion Conjecture]
	\label{conj:Main}
	There exists a positive constant $\varepsilon$ such that, for each 
	nonnegative integer $k$, we have

		\begin{equation*}
			F_N = \sum_{g=0}^k N^{2-2g} F_{Ng} + o(N^{2-2k})
		\end{equation*}
		
	\noindent
	as $N \to \infty$, where $F_N=\log I_N$ and the error term is uniform over complex numbers $z$ of modulus at most $\varepsilon$ and complex
	matrices $A,B$ of spectral radius at most $1,$ with $F_{Ng}$ analytic functions of $z$ and the eigenvalues of $A$ and $B$ whose
	modulus is uniformly bounded in $N$. Moreover, $F_{Ng}$ is a generating function for combinatorial invariants of
	compact connected genus $g$ Riemann surfaces.	
	\end{conj}

In stating Conjecture \ref{conj:Main} we have introduced the following notational convention
used throughout the paper: any declarative sentence in which the symbol $I_N$ appears without a superscript 
holds true for both \eqref{eqn:BGW} and \eqref{eqn:HCIZ}. The same convention is applied to functions of 
these integrals, e.g. $F_N=\log I_N$. The superscript is restored whenever necessary or helpful.
	
The main result of this paper is a proof of Conjecture \ref{conj:Main} --- we show that
the claimed logarithmic asymptotics of the oscillatory integrals $I_N^{(m)}$ do indeed exist,
and enumerate branched covers of the Riemann sphere 
with at most $m$ non-simple branch points. The precise statement is Theorem \ref{thm:Main} below.

The resolution of Conjecture \ref{conj:Main} presented here has many
potential applications. In particular, it should help clear the way for the development of Fourier-analytic 
techniques in the asymptotic spectral analysis of random matrices, and the parallel development of methods in asymptotic representation theory 
and integrable probability based on the orbit method. 
Moreover, the $m=1$ case should be useful for certain calculations 
in $\group{U}(N)$ lattice gauge theory, while the $m=2$ case should have
applications to multimatrix models in random matrix theory.

\subsection{Background}
The integrals \eqref{eqn:BGW} and \eqref{eqn:HCIZ} are analytic continuations of
natural integral transforms: the former was studied by James \cite{James} as the Fourier transform of Haar measure on $\group{U}(N)$,
while the latter was introduced by Harish-Chandra \cite{HC} as 
the Fourier transform of Haar measure on a given coadjoint orbit of $\group{U}(N)$.
Both are Bessel functions of matrix argument, a class of special functions
introduced by Herz \cite{Herz} whose properties at fixed rank $N$ have been studied
from a variety of perspectives \cite{Constantine,EK,GR,Macdonald:hyp}.

The $N \to \infty$ asymptotic behavior of the integrals $I_N$ was first studied 
by theoretical physicists working in quantum field theory, and it was
in this context that the first versions of Conjecture \ref{conj:Main} emerged \cite{Bars,BarsGreen,BG,GW,IZ,Samuel,Wadia}.
In particular, the hypothetical form of the asymptotic expansions claimed in Conjecture \ref{conj:Main}, 
as well as well as their supposed relation to Riemann surfaces,
derives from the interplay between two fundamental approximation schemes used in quantum field theory --- the 
strong coupling expansion \cite{Wilson} and the large $N$ expansion \cite{tHooft}.

To explain further, let us view $I_N$ as the partition function 
of a Gibbs measure on $\group{U}(N)$, with the complex variable $z$ playing the role of an inverse
coupling/temperature parameter, and the complex matrices $A$ and $B$ being viewed as external fields. 
We shall see below that $I_N^{(2)}$
depends on $A$ and $B$ only through their eigenvalues $a_1,\dots,a_N,b_1,\dots,b_N \in \C$,
while $I_N^{(1)}$ depends on $A$ and $B$ only through 
the eigenvalues of their product, so that we may take $B$ to be the identity matrix in \eqref{eqn:BGW}.
Thus $I_N^{(m)}$ may be viewed as an entire function of $mN+1$ complex variables
whose restriction to $\R^{mN+1}$ takes positive values, defining a true 
partition function in the sense of statistical physics. See \cite{LMV} for some recent algorithmic results
on sampling from the corresponding Gibbs measure in the case $m=2$.

Since $I_N=1$ at infinite coupling $z=0$, the free energy 
$F_N=\log I_N$ is defined and holomorphic in a neighborhood of the infinite coupling hyperplane 
$\{z = 0\} \subset \C^{mN+1}$. Take the closed origin-centered polydisc $\D_N$ of unit polyradius
$(1,\dots,1)$ in the phase space $\C^{mN}$ of the external field eigenvalues, embed it in the infinite coupling hyperplane,
and thicken it out to $\D_N(\varepsilon)$, the closed origin-centered polydisc of 
polyradius $(\varepsilon,1,\dots,1)$ in $\C^{mN+1}$, i.e.
the complex variables version of a compact box of height $2\varepsilon$ in the 
inverse coupling dimension and width $2$ in each of the $mN$ eigenvalue dimensions.
By construction, $F_N$ vanishes on the null set $\D_N(0) \subset \C^{mN+1}$, and
Conjecture \ref{conj:Main} posits an approximation of this holomorphic function on the full-dimensional 
compact set $\D_N(\varepsilon)$ in the regime where $N \to \infty$ with $\varepsilon > 0$ fixed.

The first claim made by Conjecture \ref{conj:Main} is that this asymptotic problem is 
well-posed, in the sense that the thickening parameter 
$\varepsilon$ in the above construction may be selected 
irrespective of the dimension parameter $N$. 
It is not at all clear that this is the case, since a priori we know nothing about the 
geometry of the zeros of the partition function:
for any given $\varepsilon > 0$, it may be that the hypersurface $\{I_N=0\}$ 
and the polydisc $\D_N(\varepsilon)$ intersect nontrivially 
in $\C^{mN+1}$ for infinitely many $N \in \N$. For example, it is known \cite{GGN3,TZ,Z} that the zero 
locus of $I_N^{(2)}$ in $\C^{2N+1}$ intersects $\D_N(\pi)=\D_N(3.14159265\dots)$ for all $N>1$,
a fact which invalidates several spurious claims in the physics literature \cite{GrossMat,KW}.
Conjecture \ref{conj:Main} is thus
predicated on the existence of an absolute constant $\delta >0$ guaranteeing nonvanishing of
$I_N$ on $\D_N(\delta)$, for all $N \in \N$. We shall refer to the hypothetical existence of 
such a constant $\delta$ as the stable non-vanishing hypothesis.

Assuming the stable non-vanishing hypothesis holds, the free energy $F_N$ belongs to the Banach algebra 
$(\O_N(\delta),\|\cdot\|_\delta)$ of germs of holomorphic functions on $\D_N(\delta)$
equipped with sup norm, for all $N \in \N$. A natural way to approximate $F_N$ is then via its Maclaurin series, 
which converges $\|\cdot\|_\delta$-absolutely.
The strong coupling expansion of $F_N$ is simply its Maclaurin series, written in the form

	\begin{equation}
	\label{eqn:StrongCouplingExpansion}
		F_N = \sum_{d=1}^\infty \frac{z^d}{d!} F_N^d,
	\end{equation}
	
\noindent
with the coefficients $F_N^d$ being viewed as homogeneous degree $d$ polynomial functions of the external field eigenvalues
(the superscript $d$ is an index, not an exponent). 

According to a fundamental principle in quantum field theory which lifts the apparatus of Feynman diagrams
to matrix integrals \cite{tHooft,BIZ,BIPZ,IZ}, and beyond \cite{EO,KS}, the strong coupling coefficients $F_N^d$ are 
expected to stratify topologically as $N \to \infty$, 

	\begin{equation}
	\label{eqn:LargeNExpansion}
		F_N^d \sim \sum_{g=0}^\infty N^{2-2g} F_{Ng}^d.
	\end{equation}

\noindent
The approximation \eqref{eqn:LargeNExpansion} is a particular instance 
of a general ansatz in quantum field theory and statistical physics 
known variously as the large $N$ expansion, the $1/N$ expansion, the genus expansion, 
the topological expansion, or the 't Hooft expansion; the leading term is called the planar 
limit, the spherical limit, or the 't Hooft limit.
The compendium \cite{BW} contains many fascinating examples of this device in action. 
In some situations, such as $\group{U}(N)$ gauge theory \cite{tHooft}
and Hermitian matrix models \cite{BIPZ},
the series \eqref{eqn:LargeNExpansion} is actually a finite sum.
This is not so in the the case at hand, which is associated with $\group{U}(N)$ 
gauge theory on a lattice --- in this case the expansion \eqref{eqn:LargeNExpansion} is always infinite and almost always divergent \cite{DtH}.

The precise meaning of the expansion \eqref{eqn:LargeNExpansion} in the present situation is that, for any fixed 
$d \in \N$ and fixed $k \in \N_0,$ we should have

	\begin{equation}
	\label{eqn:LargeNExpansionPrecise}
		\lim_{N \to \infty} N^{2k-2} \left\| F_N^d - \sum_{g=0}^k N^{2-2g} F_{Ng}^d \right\| = 0,
	\end{equation}
	
\noindent
where $\|\cdot\|$ is sup norm on bounded functions on the unit polydisc $\D_N \subset \C^{mN}$,
and $F_{Ng}^d$ is a homogeneous degree $d$ polynomial in $mN$ variables which satisfies

	\begin{equation}
	\label{eqn:Uniform}
		\sup_{N \in \N} \|F_{Ng}^d\| < \infty.
	\end{equation}
	
\noindent
The topological feature of the polynomials $F_{Ng}^d$ alluded to above 
is that, when expressed as polynomials in the moments of the empirical eigenvalue
distributions of the external fields, their coefficients are expected to be topological
invariants of compact connected genus $g$ Riemann surfaces. 
Whereas the topological invariants underlying the large $N$ expansion 
in continuum gauge theory and Hermitian matrix models have long been known to be counts of
isotopy classes of embedded graphs, 
a fact which has had many ramifications in mathematical physics \cite{DGZ,Witten} and
algebraic geometry \cite{Looijenga}, the topological invariants 
which presumably underly the large $N$ expansion in nonabelian lattice gauge theories have remained
mysterious. It has even been argued that the leading order $F_{N0}^d$ in
\eqref{eqn:LargeNExpansion} does not consist of purely genus zero information \cite{Weingarten}.
See \cite{OZ2} for further discussion of these issues. 

We will see below that the large $N$ expansion \eqref{eqn:LargeNExpansion} of $F_N^d$ does indeed exist and 
that its coefficients $F_{Ng}^d$ admit a topological interpretation in terms of degree $d$ branched covers of
the Riemann sphere by a compact connected genus $g$ surface, with at most $m$ non-simple branch points. 
We will also see that
the series \eqref{eqn:LargeNExpansion} converges uniformly absolutely 
on compact subsets of $\C^{mN}$ for $d\leq N$ (the stable range)
but not for $d>N$ (the unstable range).
Conjecture \ref{conj:Main} expresses the hope that one can nevertheless plug the
large $N$ expansion \eqref{eqn:LargeNExpansion}
into the strong coupling expansion \eqref{eqn:StrongCouplingExpansion}, 

	\begin{equation}
		F_N=\sum_{d=1}^\infty \frac{z^d}{d!} F_N^d \longrightarrow \sum_{d=1}^\infty \frac{z^d}{d!} \sum_{g=0}^\infty N^{2-2g} F_{Ng}^d,
	\end{equation}

\noindent
change order of summation, 

	\begin{equation}
		\sum_{d=1}^\infty \frac{z^d}{d!} \sum_{g=0}^\infty N^{2-2g} F_{Ng}^d \longrightarrow 
		\sum_{g=0}^\infty N^{2-2g} \sum_{d=1}^\infty \frac{z^d}{d!} F_{Ng}^d,
	\end{equation}

\noindent
and thereby arrive at an $N \to \infty$ asymptotic expansion of the free energy itself, 

	\begin{equation}
	\label{eqn:TopologicalExpansion}
		F_N \sim \sum_{g=0}^\infty N^{2-2g} F_{Ng},
	\end{equation}

\noindent
the coefficients of which are genus-specific generating functions,

	\begin{equation}
	\label{eqn:GenusSpecific}
		F_{Ng} = \sum_{d=1}^\infty \frac{z^d}{d!}F_{Ng}^d.
	\end{equation}
	
\noindent
Conjecture \ref{conj:Main} is thus claiming the existence of three
absolute constants: a $\delta > 0$ such that the integral $I_N$ is non-vanishing on $\D_N(\delta)$ 
for all $N$, a $\gamma > 0$ such that 
the series $F_{Ng}$ converges uniformly absolutely on $\D_N(\gamma)$ for all $N$ and $g$,
and finally an $\varepsilon \in (0,\delta) \cap (0,\gamma)$, such that

	\begin{equation}
	\label{eqn:TopologicalLimit}
		\lim_{N \to \infty} N^{2k-2} \left\| F_N - \sum_{g=0}^k N^{2-2g} F_{Ng} \right\|_\varepsilon = 0
	\end{equation}

\noindent
holds for any fixed $k$. 

The exchangeability of the large $N$ expansion and the strong coupling expansion 
has never been verified, and historically has been viewed
with suspicion \cite{GW}. Instead, asymptotic studies of $F_N=\log I_N$
have proceeded along different lines, employing variational calculus \cite{Matytsin},
large deviation theory \cite{GZ}, asymptotics of Toeplitz 
determinants \cite{Johansson:MRL}, and Schwinger-Dyson ``loop'' equations
\cite{BG,CGM,GN}. These approaches all have a common shortcoming: they only work when 
$I_N$ is restricted to the real domain $\R^{mN+1} \subset \C^{mN+1}$, so that oscillatory 
behavior of $I_N$ is suppressed and probabilistic estimates can be used.
The restriction to real parameters is a substantial drawback, since as mentioned
above the BGW and HCIZ integrals are natural Fourier kernels for Haar unitary random matrices and unitarily invariant
Hermitian random matrices, respectively, which forces the issue of understanding their large $N$ behavior 
for complex coupling and complex external fields. While Fourier methods in random matrix theory and integrable probability 
are in active development \cite{BGH,BufGor,CMZ,Faraut,KR,MN3,Zuber:Horn}, a persistent obstacle 
has been that only non-oscillatory asymptotics for the fundamental kernels \eqref{eqn:BGW} 
and \eqref{eqn:HCIZ} have so far been available.

\subsection{Result}
The main result of this paper is a proof of Conjecture \ref{conj:Main}.
Throughout, $(\O_N(\rho),\|\cdot\|_\rho)$ denotes the Banach algebra of germs of
holomorphic functions on the closed origin-centered polydisc of polyradius 
$(\rho,1,\dots,1)$ in $\C^{mN+1}$, equipped with uniform norm. 
The unadorned norm $\|\cdot\|$ always means sup norm on bounded functions 
on the closed unit polydisc $\D_N$ in eigenvalue phase space $\C^{mN}$.
Given a Young diagram $\alpha$, we write $\alpha \vdash d$ to indicate that
$\alpha$ consists of $d$ cells. For a diagram $\alpha$ with $\ell(\alpha)$ rows
and $\alpha_i$ cells in the $i$th row, let 

	\begin{equation}
		p_\alpha(x_1,\dots,x_N) = \prod_{i=1}^{\ell(\alpha)} \sum_{j=1}^N x_j^{\alpha_i}
	\end{equation}
	
\noindent
denote the corresponding Newton symmetric polynomial in $N$ variables. Our main result is the following.
	
	\begin{thm}
	\label{thm:Main}
	There exists $\varepsilon \in (0,\frac{2}{27})$ such that $I_N$ is nonvanishing
	on $\D_N(\varepsilon)$ for all $N \in \N$. Letting
	$F_N=\log I_N \in \O_N(\varepsilon)$, for each $k \in \N_0$ we have
		
		\begin{equation*}
			\lim_{N \to \infty} N^{2k-2} \bigg{\|} F_N - \sum_{g=0}^k N^{2-2g} F_{Ng} \bigg{\|}_\varepsilon = 0,
		\end{equation*}
		
	\noindent
	where

			\begin{equation*}
			\begin{split}
				F_{Ng}^{(1)} & = \sum_{d=1}^\infty  \frac{z^d}{d!} \sum_{\alpha \vdash d}
				\frac{p_{\alpha}(a_1,\dots,a_N)}{N^{\ell(\alpha)}} (-1)^{\ell(\alpha)+d} \mon_g(\alpha), \\
				F_{Ng}^{(2)} &= \sum_{d=1}^\infty  \frac{z^d}{d!} \sum_{\alpha,\beta \vdash d}
				\frac{p_\alpha(a_1,\dots,a_N)}{N^{\ell(\alpha)}} \frac{p_\beta(b_1,\dots,b_N)}{N^{\ell(\beta)}} (-1)^{\ell(\alpha)+\ell(\beta)} 
				\mon_g(\alpha,\beta),
			\end{split}
			\end{equation*}
		
	\noindent
	and the positive integers $\mon_g(\alpha)$,  $\mon_g(\alpha,\beta)$ are, respectively, the connected monotone 
	single and double Hurwitz numbers of degree $d$ and genus $g$.
	\end{thm}

Monotone Hurwitz numbers, introduced in \cite{GGN1,GGN2} and further studied in many papers since, 
are combinatorial variants of the classical Hurwitz numbers familiar from enumerative geometry
as counts of branched covers of the Riemann sphere with specified ramification data.
The enumerative study of branched covers was initiated by Hurwitz in the 19th century \cite{Hurwitz1,Hurwitz2}, 
and is now known as Hurwitz theory.  Hurwitz theory remains relevant in contemporary enumerative geometry, 
being deeply intertwined with more sophisticated approaches to the enumeration of maps from curves to curves \cite{OP}.
Hurwitz theory is the subject of a huge literature; see \cite{EEHS} for a quick introduction,
\cite{CM} for a pedagogical treatment, \cite{Lando:ICM} for a survey of connections to other fields,
and \cite{DYZ} for a more recent reference whose perspective aligns closely with that of the present paper.

Monotone Hurwitz numbers are defined precisely in Section \ref{sec:Finite} below. For now, suffice to say that
they are obtained from classical Hurwitz numbers via a combinatorial desymmetrization, rooted in the representation
theory of the symmetric groups, which leaves all the main structural features of Hurwitz
theory intact. While monotone Hurwitz numbers are not as geometrically natural as their 
classical counterparts, 
they have a major quantitative advantage: they are smaller. Consequently, the asymptotic behavior 
of monotone Hurwitz numbers is quite tame, and their generating functions have robust summability properties \cite{GGN5}, 
making them well-behaved as analytic objects. Monotone Hurwitz theory provides a natural Feynman diagram 
apparatus for $\group{U}(N)$ integrals \cite{Novak:Banach}, yielding a useful
description of the large $N$ expansion in $\group{U}(N)$ lattice gauge theories.

Although interest in monotone Hurwitz numbers has exploded in recent years, expanding the scope of
Hurwitz theory \cite{ACEH,ALS,BGF,CD} and leading to new connections with between enumerative
geometry and matrix models \cite{BK,CDO,GGR1,GGR2,Montanaro}, 
they were originally conceived in order to address the HCIZ case of Conjecture \ref{conj:Main}, with \cite{GGN3} taking the first steps in this direction.
In this paper, the program begun in \cite{GGN3} achieves a considerably enhanced fulfillment of its initial purpose.
	

\section{Finite $N$}
\label{sec:Finite}

	In this section, we analyze the integrals \eqref{eqn:BGW} and \eqref{eqn:HCIZ} with $N \in \N$ arbitrary but fixed,
	and obtain their coupling expansions: absolutely convergent power series in the inverse coupling parameter $z$ whose coefficients are symmetric 
	polynomials in the eigenvalues of the external fields $A$ and $B$. These coupling expansions may be presented 
	either in terms of Schur polynomials (character form), or Newton polynomials (string form). 
	The character form is widely known, but it is from the string form that a link with
	enumerative geometry emerges.
								
	\subsection{Character form}
	Given a Young diagram $\lambda$, let $s_\lambda(x_1,\dots,x_N)$ be the 
	corresponding Schur polynomial in $N$ variables \cite{Macdonald,Stanley:EC2}.	
	The evaluation $s_\lambda(A)$ of $s_\lambda$ on the spectrum of $A \in \group{GL}_N(\C)$ is 
	the character $\Tr S^\lambda(A)$ of $A$ acting in the irreducible polynomial representation $(\mathsf{W}^\lambda,S^\lambda)$
	of the general linear group indexed by $\lambda$.
	The following integration formulas are standard consequences of Schur orthogonality
	together with a density argument; see \cite{Macdonald}.
	
		\begin{lem}
			\label{lem:BasicIntegrals}
			For any Young diagrams $\lambda,\mu$ with at most $N$ rows and any 
			matrices $A,B \in \mathfrak{gl}_N(\C)$, we have				
			
				\begin{equation}
					\int_{\group{U}(N)} s_\lambda(AU) s_\mu(BU^{-1}) \mathrm{d}U
					= \delta_{\lambda\mu}\frac{s_\lambda(AB)}{\dim \mathsf{W}^\lambda}.
				\end{equation}
				
			\noindent
			and
				
				\begin{equation}
					\int_{\group{U}(N)} s_\lambda(AUBU^{-1}) \mathrm{d}U
					= \frac{s_\lambda(A) s_\lambda(B)}{\dim \mathsf{W}^\lambda}
				\end{equation}

		\end{lem}

		Lemma \ref{lem:BasicIntegrals} leads directly to the character form of the coupling 
		expansions of the BGW and HCIZ integrals. These character expansions are widely 
		known in physics, see e.g. the reviews \cite{Morozov,ZZ}, and seem to have first appeared explicitly in work of James \cite{James} 
		in multivariate statistics. Let $(\mathrm{V}^\lambda,R^\lambda)$ denote
		the irreducible representation of the symmetric group $\group{S}(d)$ associated to a Young diagram 
		$\lambda \vdash d$, and let 
		
			\begin{equation}
				\chi^\lambda_\alpha = \Tr R^\lambda(\pi)
			\end{equation}
			
		\noindent
		be the character of a permutation $\pi$ from the conjugacy class $C_\alpha \subset \group{S}(d)$ 
		acting in $\mathrm{V}^\lambda$.
											
		\begin{thm}
			\label{thm:CharacterForm}
			For any $A,B \in \mathfrak{gl}_N(\C)$, the integrals 
			\eqref{eqn:BGW} and \eqref{eqn:HCIZ} admit power series expansions
						
				\begin{equation*}
					I_N^{(1)} = 1 + \sum_{d=1}^\infty  \frac{z^d}{d!} I_N^{(1)d} \quad\text{ and }\quad
					I_N^{(2)} = 1 + \sum_{d=1}^\infty \frac{z^d}{d!} I_N^{(2)d}
				\end{equation*}
				
			\noindent
			which converge absolutely for all $z \in \C$, and whose coefficients are given by
					
				\begin{equation*}
				I_N^{(1)d} = \frac{N^{2d}}{d!} 
				 \sum_{\substack{\lambda \vdash d \\ \ell(\lambda) \leq N}} 
				s_\lambda(AB) \frac{(\dim \mathsf{V}^\lambda)^2}{\dim \mathsf{W}^\lambda}
				 \quad\text{ and }\quad
				I_N^{(2)d} = N^d
				 \sum_{\substack{\lambda \vdash d \\ \ell(\lambda) \leq N}} 
				 s_\lambda(A)s_\lambda(B) \frac{\dim \mathsf{V}^\lambda}{\dim \mathsf{W}^\lambda}.		
				\end{equation*}
				
			\end{thm}
						
		\begin{proof}	
		For any $A,B \in \mathfrak{gl}_N(\C)$, the integral $I_N$
		is an entire functions of $z$ whose derivatives may be computed by differentiating 
		under the integral sign. 
		
		Using the first integration formula in Lemma \ref{lem:BasicIntegrals}, 
		the Maclaurin series of $I_N^{(1)}$ is 
				
			\begin{equation}
			\begin{split}
				 I_N^{(1)} &=  1 + \sum_{d=1}^\infty \frac{z^d}{d!d!} N^{2d}
				\sum_{\substack{\lambda \vdash d \\ \ell(\lambda) \leq N}} 
				\sum_{\substack{\mu \vdash d \\ \ell(\mu) \leq N}} 
				(\dim \mathsf{V}^\lambda)(\dim V^\mu)
				\int_{\group{U}(N)} s_\lambda(AU) s_\mu(BU^{-1}) \mathrm{d}U   \\
				&= 1 + \sum_{d=1}^\infty \frac{z^d}{d!} \frac{N^{2d}}{d!}
				\sum_{\substack{\lambda \vdash d \\ \ell(\lambda) \leq N}} 
				s_\lambda(AB)\frac{(\dim \mathsf{V}^\lambda)^2}{\dim \mathsf{W}^\lambda}.
			\end{split}
			\end{equation}	

		For $I_N^{(2)},$ we first write
		
			\begin{equation}
				I_N^{(2)} = 1 + \sum_{d=1}^\infty \frac{z^d}{d!} N^d \int_{\group{U}(N)} (\Tr AUBU^{-1})^d \mathrm{d}U,
			\end{equation}
			
		\noindent
		 and observe that 
		
			\begin{equation}
				 (\Tr AUBU^{-1})^d = p_{(1^d)}(AUBU^{-1}),
			\end{equation}
		
		\noindent
		where $p_{(1^d)}$ is the Newton symmetric polyomial indexed by the columnar diagram with 
		$d$ cells. For any $\alpha \vdash d$, we have
		
			\begin{equation}
				p_\alpha(x_1,\dots,x_N) = \sum_{\lambda \vdash d} \chi^\lambda_\alpha s_\lambda(x_1,\dots,x_N),
			\end{equation}
			
		\noindent
		and in particular
		
			\begin{equation}
				p_{(1^d)}(x_1,\dots,x_N) = \sum_{\lambda \vdash d} (\dim \mathsf{V}^\lambda) s_\lambda(x_1,\dots,x_N).
			\end{equation}
		
		\noindent
		Noting that $s_\lambda(x_1,\dots,x_N)$ is the zero polynomial if $\ell(\lambda)>N,$ we thus have that
		
			\begin{equation}
			\begin{split}
				\int_{\group{U}(N)} p_{(1^d)}(AUBU^{-1}) \mathrm{d}U &= N^d
				\sum_{\substack{\lambda \vdash d \\ \ell(\lambda) \leq N}}
				(\dim \mathsf{V}^\lambda) \int_{\group{U}(N)} s_\lambda(AUBU^{-1}) \mathrm{d}U \\
				&=  N^d
				\sum_{\substack{\lambda \vdash d \\ \ell(\lambda) \leq N}}
				s_\lambda(A) s_\lambda(B)\frac{\dim \mathsf{V}^\lambda}{\dim \mathsf{W}^\lambda},
			\end{split}
			\end{equation}

		\noindent
		by the second integration formula in Lemma \ref{lem:BasicIntegrals}. 		
		\end{proof}

	Note that the adjective ``strong'' is not necessary here, since these 
	coupling expansions are globally convergent. It is clear that
	we may take $B$ to be the identity
	in the definition of $I_N^{(1)}$ without loss in generality,
	and we do so going forward.
		
	\subsection{Strong coupling approximation}
	An important consequence of Theorem \ref{thm:CharacterForm} is that $I_N$ is exponentially well approximated by
	its coupling expansion to $\Theta(N^2)$ terms, provided the external fields
	have spectral radius bounded independently of $N$ and the 
	coupling is sufficiently strong. More precisely, given $t>0$ let 
	
		\begin{equation}
			P_N = 1 + \sum_{d=1}^{\lfloor tN^2 \rfloor} \frac{z^d}{d!} I_N^d
		\end{equation}
		
	\noindent
	be the truncation of the coupling expansion of $I_N$ at degree $d=\lfloor tN^2 \rfloor$.
	We omit the dependence of $P_N$ on the truncation parameter $t$ in order
	to lighten the notation; in Section \ref{sec:Large}, we will fix $t$ at a specific value.
	Define
	
		\begin{equation}
			u(x) = \frac{1}{1-\frac{ex}{t}} \quad\text{ and } \quad
			v(x) = t \log \frac{t}{ex},
		\end{equation}
		
	\noindent
	and observe that 
	
		\begin{equation}
			\lim_{x \to 0^+}u(x) = 1 \quad\text{ and }\quad
			\lim_{x \to 0^+}v(x)  = \infty.
		\end{equation}
		
	\noindent
	We then have the following comparison inequality.
	
		\begin{thm}
		\label{thm:StrongCouplingApproximation}
		For any $\rho < \frac{t}{e},$ we have
			
			\begin{equation*}
				\|I_N-P_N\|_\rho < u(\rho)e^{-v(\rho)N^2}.
			\end{equation*}
			
		\end{thm}
		
		\begin{proof}
			Since the Schur polynomials are monomial positive, we have
			
				\begin{equation}
					\|s_\lambda(x_1,\dots,x_N)\|=s_\lambda(1,\dots,1) = \dim \mathsf{W}^\lambda,
				\end{equation}
				
			\noindent
			and hence
			
				\begin{equation}
					\|I_N^{(1)}\| =\frac{N^{2d}}{d!} 
					\sum_{\substack{\lambda \vdash d \\ \ell(\lambda) \leq N}}
					s_\lambda(1,\dots,1)\frac{(\dim \mathsf{V}^\lambda)^2}{\dim \mathsf{W}^\lambda} \\
					= \frac{N^{2d}}{d!} 
					\sum_{\substack{\lambda \vdash d \\ \ell(\lambda) \leq N}}
					(\dim \mathsf{V}^\lambda)^2.
				\end{equation}
				
			\noindent
			The function 
			
				\begin{equation}
					\lambda \mapsto \frac{(\dim \mathsf{V}^\lambda)^2}{d!}
				\end{equation}
				
			\noindent
			is the mass function of the the Plancherel measure, i.e. the 
			probability measure on Young diagrams canonically associated
			to the Fourier isomorphism
			
				\begin{equation}
					\C\group{S}(d) \longrightarrow \bigoplus_{\lambda \vdash d} \End \mathrm{V}^\lambda.
				\end{equation}
				
			\noindent
			Consequently, we have
			
				\begin{equation}
					\|I_N^{(1)}\| \leq N^{2d},
				\end{equation}
				
			\noindent
			with equality if and only if $d \leq N$.

			In the case $m=2$, the same argument gives
			
				\begin{equation}
				\begin{split}
					\|I_N^{(2)d}\|&= N^d 
					\sum_{\substack{\lambda \vdash d \\ \ell(\lambda) \leq N}}
					s_\lambda(1,\dots,1) s_\lambda(1,\dots,1)\frac{\dim \mathsf{V}^\lambda}{\dim \mathsf{W}^\lambda} \\
					&= N^d \sum_{\substack{\lambda \vdash d \\ \ell(\lambda) \leq N}}
					\dim \mathsf{V}^\lambda \dim \mathsf{W}^\lambda \\
					&=  N^{2d}, 
				\end{split}
				\end{equation}
		
			\noindent
			where the final equality follows from the Schur-Weyl isomorphism \cite{Macdonald}
	
			\begin{equation}
				\bigoplus\limits_{\substack{\lambda \vdash d \\ \ell(\lambda) \leq N}}
				\mathsf{V}^\lambda \otimes \mathsf{W}^\lambda \longrightarrow (\C^N)^{\otimes d}.
			\end{equation}
			
			Since 
											
			\begin{equation}
				\|I_N^{(1)}\| \leq \|I_N^{(2)}\| = N^{2d},
			\end{equation}
			
		\noindent
		we have the bound
						
			\begin{equation}
				\|I_N-P_N\|_\rho \leq \sum_{d> tN^2} \frac{\rho^d}{d!} N^{2d}.
			\end{equation}
			
		\noindent
		The result thus follows from the elementary estimate
				
			\begin{equation}
				d! > \frac{d^d}{e^d},
			\end{equation}		
			
		\noindent
		which for $d > tN^2$ gives
		
			\begin{equation}
				\frac{\rho^d}{d!} N^{2d} <(\rho e)^d \left( \frac{N^2}{d} \right)^d 
				< \left( \frac{\rho e}{t} \right)^d,
			\end{equation}
			
		\noindent
		so that for $\rho < \frac{t}{e}$ we have
		
			\begin{equation}
				\sum_{d> tN^2} \frac{\rho^{d}}{d!} N^{2d} <\left( \frac{\rho e}{t} \right)^{\lfloor tN^2 \rfloor +1}
				\sum_{k=0}^\infty  \left( \frac{\rho e}{t} \right)^k < \frac{1}{1-\frac{\rho e}{t}}
				\left( \frac{\rho e}{t} \right)^{ tN^2}.
			\end{equation}
		\end{proof}
		
	\begin{remark} 
	The argument above shows that in the absence of external fields the HCIZ integral 
	$I_N^{(1)}$ degenerates to the exponential function
	
		\begin{equation}
			E_N = 1 + \sum_{d=1}^\infty \frac{z^d}{d!} N^{2d}.
		\end{equation}
		
	\noindent
	This is clear from the integral representation \eqref{eqn:HCIZ}, and the proof
	of Theorem \ref{thm:StrongCouplingApproximation} reproduces this obvious fact algebraically.
	On the other hand, the fieldless BGW integral
	
		\begin{equation}
			\label{fieldless:BGW}
				L_N= \int_{\mathrm{U}(N)} e^{\sqrt{z}N \mathrm{Tr}(U + U^{-1})} \mathrm{d}U
		\end{equation}

	\noindent
	and its coupling expansion 
	
		\begin{equation}
			L_N = 1 + \sum_{d=1}^\infty \frac{z^d}{d!} N^{2d}\sum_{\substack{\lambda \vdash d \\ \ell(\lambda) \leq N}}
				\frac{(\dim \mathsf{V}^\lambda)^2}{d!}
		\end{equation}
		
	\noindent
	are non-trivial objects of substantial interest.
	Via the Robinson-Schensted correspondence \cite{Stanley:EC2},
	we have
	
		\begin{equation}
			\sum_{\substack{\lambda \vdash d \\ \ell(\lambda) \leq N}}
				\frac{(\dim \mathsf{V}^\lambda)^2}{d!}= N^{2d} \P[\LIS_d \leq N],
		\end{equation}
		
	\noindent
	where $\P[\LIS_d \leq N]$ is the probability that a uniformly random 
	permutation from $\group{S}(d)$ has maximal increasing subsequence length
	at most $N$, and thus
	
		\begin{equation}
		\label{eqn:Rains}
			L_N = 1 + \sum_{d=1}^\infty \frac{z^d}{d!} N^{2d}\P[\LIS_d \leq N].
		\end{equation}
		
	\noindent
	The remarkable fact that the coupling expansion of the 
	fieldless BGW integral is a generating function for the cumulative distribution 
	function of $\LIS_d$ was first explicitly recognized by Rains \cite{Rains}; see \cite{BR} for extensions to
	other groups, and \cite{Novak:EJC,Novak:IMRN} for generalizations to integrals over truncated unitary matrices.
	Via the Heine-Szeg\H{o} formula \cite{Bump}, the expansion \eqref{eqn:Rains} is equivalent to a Toeplitz determinant representation 
	of this series earlier derived by Gessel \cite{Gessel}, without mentioning matrix integrals. Interestingly, the determinant representation 
	of the BGW integral was known to the physicists Bars and Green earlier still \cite{BarsGreen}, though without 
	the connection to Plancherel measure or increasing subsequences. The power series
	expansion \eqref{eqn:Rains} of the the fieldless BGW integral was utilized by Johansson
	\cite{Johansson:MRL} and Baik-Deift-Johansson \cite{BDJ} to analyze the $d \to \infty$ behavior of the random variable 
	$\LIS_d$. We shall return to this in Section \ref{sec:Large} below. For more on the
	fascinating topic of longest increasing subsequences in random permutations, see \cite{AD,Deift,Romik,Stanley:ICM}.
	\end{remark}
		
	\subsection{Another one}
	Via Theorem \ref{thm:CharacterForm}, we may view $I_N^{(m)}$, $m=1,2$,
	as entire functions on $\C^{mN+1}$ defined by globally convergent power series. 
	This will be our point of view going forward, and the initial definitions of 
	these functions as matrix integrals will play no further role in our analysis.
	Without further ado, we introduce another series $I_N^{(0)}$ 
	which fits naturally with $I_N^{(1)}$ and $I_N^{(2)}$, but for which we 
	claim no matrix integral representation.
	
	Recall the standard dimension formulas \cite{Macdonald,Stanley:EC2}
	
		\begin{equation}
		\label{eqn:DimensionFormulas}
			\dim \mathsf{V}^\lambda = \frac{d!}{\prod_{\Box \in \lambda} h(\Box)}
			\quad\text{ and }\quad 
			\dim \mathsf{W}^\lambda = \prod_{\Box \in \lambda} \frac{N+c(\Box)}{h(\Box)},
		\end{equation}
		
	\noindent
	where $c(\Box)$ and $h(\Box)$ denote the content and hook length, respectively,
	of a given cell $\Box \in \lambda$. Applying these formulas, the coupling expansions
	of the HCIZ and BGW integrals in character form become
		
		\begin{equation}
		\begin{split}
			I_N^{(1)} &= 1+\sum_{d=1}^\infty z^d N^d \sum_{\substack{\lambda \vdash d \\ \ell(\lambda) \leq N}}
			s_\lambda(a_1,\dots,a_N) \prod_{\Box \in \lambda} \frac{1}{h(\Box)(1+\frac{c(\Box)}{N})} \\
			I_N^{(2)} &= 1+ \sum_{d=1}^\infty z^d \sum_{\substack{\lambda \vdash d \\ \ell(\lambda) \leq N}}
			s_\lambda(a_1,\dots,a_N) s_\lambda(b_1,\dots,b_N) \prod_{\Box \in \lambda} \frac{1}{1+\frac{c(\Box)}{N}}
		\end{split}
		\end{equation}
		
	\noindent
	Extrapolating, we define the univariate entire function 
	
		\begin{equation}
			I_N^{(0)} = 1+\sum_{d=1}^\infty z^d N^{2d}\sum_{\substack{\lambda \vdash d \\ \ell(\lambda) \leq N}} \prod_{\Box \in \lambda} \frac{1}{h(\Box)^2(1+\frac{c(\Box)}{N})}.
		\end{equation}
		
	\noindent
	As in the cases $m=1,2$, we denote by $I_N^{(0)d}$ the coefficient of $z^d/d!$ in the series $I_N^{(0)}$.
	We extend the practice of referring to any/all of these three functions as simply 
	$I_N$ in statements which hold uniformly for $I_N^{(m)}$, $m \in \{0,1,2\}.$

	\subsection{String form}	
	Newton polynomials are much more elementary objects than Schur polynomials:
	when evaluated on matrix eigenvalues they are trace invariants rather than characters. Thus, 
	by expressing the coupling coefficients of $I_N$ in terms of Newton polynomials 
	rather than Schur polynomials we will be expressing them in terms of the moments
	of the empirical eigenvalue distributions of the external fields.
	Newton polynomials are the preferred system of 
	symmetric polynomials for coupling expansions in gauge theory, where they are referred 
	to as string states \cite{BT}.
	
	The string form of $I_N^d$ involves the central characters of the symmetric group 
	$\group{S}(d)$, whose definition we now recall.
	Given a diagram $\alpha \vdash d$, we identify the corresponding conjugacy class 
	$C_\alpha$ in the symmetric group $\group{S}(d)$ with the formal sum of its elements, so that
	it becomes an element of the group algebra $\mathbb{C}\group{S}(d)$ and
	$\{C_\alpha \colon \alpha \vdash d\}$, is a linear basis for the center $\mathcal{Z}(d)$ of
	$\C\group{S}(d)$. By Schur's Lemma, $C_\alpha$ acts as a scalar operator in any irreducible 
	representation $(\mathsf{V}^\lambda,R^\lambda)$ of $\mathbb{C}\group{S}(d)$:
	we have
	
		\begin{equation}
			R^\lambda(C_\alpha) = \omega_\alpha(\lambda) I_{\mathsf{V}^\lambda}
		\end{equation}
		
	\noindent
	where 
	
		\begin{equation}
			\omega_\alpha(\lambda) = \frac{|C_\alpha| \chi^\lambda_\alpha}{\dim \mathsf{V}^\lambda}
		\end{equation}
		
	\noindent
	and $I_{\mathsf{V}^\lambda} \in \mathrm{End} \mathsf{V}^\lambda$ is the identity operator. 
	The traditional representation-theoretic perspective is to view the central characters as functions
	of the representation, i.e. as the homomorphisms $\mathcal{Z}(d) \longrightarrow \C$ which 
	send each central element to its eigenvalue acting in a given irreducible representation. However, in 
	the representation theory of the symmetric group it is useful to invert this perspective and think of central
	characters as conjugacy class indexed functions of the representation.
	It is an important theorem of Kerov and Olshanski \cite{OV} that $\omega_\alpha(\lambda)$ is a symmetric function
	of the contents of $\lambda$; see also \cite{CGS,Olshanski}.
	
	In addition to the central characters, we need a second family
	$\Omega_\hbar$ of functions on Young diagrams indexed by a parameter $\hbar \in \C$. These are defined
	explicitly as 
	
		\begin{equation}
		\label{eqn:OmegaPoints}
			\Omega_\hbar(\lambda) = \prod_{\Box \in \lambda} (1+\hbar c(\Box)).
		\end{equation}
		
	\noindent
	The functions $\Omega_\hbar^{\pm 1}$ play an important role in 2D Yang-Mills theory, where they are referred to as
	Omega points \cite{CMR,GT2}. 
											
	\begin{thm}
	\label{thm:StringForm}
		For each $d \in \N$, we have
		
			\begin{equation*}
			\begin{split}
				I_N^{(0)d} &= N^{2d} \sum_{\substack{\lambda \vdash d \\ \ell(\lambda) \leq N}}
					\Omega_{\frac{1}{N}}^{-1}(\lambda)  \frac{(\dim \mathsf{V}^\lambda)^2}{d!},\\
				I_N^{(1)d} &= N^d \sum_{\alpha \vdash d} p_\alpha(a_1,\dots,a_N) \sum_{\substack{ \lambda \vdash d \\ \ell(\lambda) \leq N}} \omega_\alpha(\lambda) 
				\Omega_{\frac{1}{N}}^{-1}(\lambda) \frac{(\dim \mathsf{V}^\lambda)^2}{d!}, \\
				I_N^{(2)d} &= \sum_{\alpha,\beta \vdash d} p_\alpha(a_1,\dots,a_N) p_\beta(b_1,\dots,b_N)
				\sum_{\substack{ \lambda \vdash d \\ \ell(\lambda) \leq N}}\omega_\alpha(\lambda) \Omega_{\frac{1}{N}}^{-1}(\lambda) \omega_\beta(\lambda)
				\frac{(\dim \mathsf{V}^\lambda)^2}{d!}.
			\end{split}
			\end{equation*}
		
	\end{thm}
	
	\begin{proof}	
	In the case $m=0$, we immediately have
	
		\begin{equation}
			I_N^{(0)d} = d! N^{2d}\sum_{\substack{\lambda \vdash d \\ \ell(\lambda) \leq N}} \prod_{\Box \in \lambda} \frac{1}{h(\Box)^2(1+\frac{c(\Box)}{N})}
			= N^{2d} \sum_{\substack{\lambda \vdash d \\ \ell(\lambda) \leq N}}
					\Omega_{\frac{1}{N}}^{-1}(\lambda)  \frac{(\dim \mathsf{V}^\lambda)^2}{d!},
		\end{equation}
		
	\noindent
	as claimed. For the cases $m=1,2,$ from the dimension formulas
	\eqref{eqn:DimensionFormulas} we have
		
		\begin{equation}
			\frac{N^d}{d!} \frac{\dim \mathsf{V}^\lambda}{\dim \mathsf{W}^\lambda} = \Omega_{\frac{1}{N}}^{-1}(\lambda).
		\end{equation}
		
	\noindent	
	From Theorem \ref{thm:CharacterForm}, we thus have
	
			\begin{align*}
				I_N^{(1)d} &= N^d\sum_{\substack{\lambda \vdash d \\ \ell(\lambda) \leq N}} 
				s_\lambda(A) \Omega_{\frac{1}{N}}^{-1}(\lambda) \dim \mathsf{V}^\lambda \\
				&=  N^d\sum_{\substack{\lambda \vdash d \\ \ell(\lambda) \leq N}} \left(
				\sum_{\alpha \vdash d} \frac{|C_\beta|\chi_\alpha(\lambda)}{d!}
				p_\alpha(A) \right)
				\Omega_{\frac{1}{N}}^{-1}(\lambda) \dim \mathsf{V}^\lambda \\
				&= N^d\sum_{\alpha \vdash d} p_\alpha(A)
				\sum_{\substack{\lambda \vdash d \\ \ell(\lambda) \leq N}}
				\omega_\alpha(\lambda) \Omega_{\frac{1}{N}}^{-1}(\lambda)\frac{(\dim \mathsf{V}^\lambda)^2}{d!},
			\end{align*}
	
	\noindent
	and
	
			\begin{align*}
				I_N^{(2)d} &= d!\sum_{\substack{\lambda \vdash d \\ \ell(\lambda) \leq N}}
				s_\lambda(A) s_\lambda(B) \Omega_{\frac{1}{N}}^{-1}(\lambda)  \\
				&= d!\sum_{\substack{\lambda \vdash d \\ \ell(\lambda) \leq N}}
				\left( \sum_{\alpha \vdash d} \frac{|C_\alpha| \chi_\alpha(\lambda)}{d!} p_\alpha(A) \right)
				\left( \sum_{\beta \vdash d} \frac{|C_\beta| \chi_\beta(\lambda)}{d!} p_\beta(B) \right) \Omega_{\frac{1}{N}}^{-1}(\lambda)  \\
				&= \sum_{\alpha,\beta \vdash d} p_\alpha(A) p_\beta(B)
				\sum_{\substack{\lambda \vdash d \\ \ell(\lambda) \leq N}} 
				\omega_\alpha(\lambda)  \Omega_{\frac{1}{N}}^{-1}(\lambda)\omega_\beta(\lambda)
				\frac{(\dim \mathsf{V}^\lambda)^2}{d!}.
			\end{align*}
	\end{proof}
	
	The string form of the coupling coefficients reveals a feature that is not obvious from the character form: 
	$I_N^{(m-1)d}$ is obtained from $I_N^{(m)d}$ by extraction and specialization. To obtain the 
	polynomial $I_N^{(1)d}$, extract the coefficient of
	
		\begin{equation}
			p_{(1^d)}(b_1,\dots,b_N) = (b_1+\dots + b_N)^d,
		\end{equation}
		
	\noindent
	in $I_N^{(2)d}$ and set $b_i=1$.
	To obtain the number $I_N^{(0)d}$ from the polynomial $I_N^{(1)d}$, extract all
	multiples of 
	
		\begin{equation}
			p_{(1^d)}(a_1,\dots,a_N) = (a_1+\dots + a_N)^d
		\end{equation}
		
	\noindent
	and set $a_i=1$.

	\subsection{Statistical expansion}
	A basic feature of the Plancherel measure is that the corresponding
	expectation functional 
	
		\begin{equation}
			\langle X \rangle = \sum_{\lambda \vdash d} X(\lambda) 
			\frac{(\dim \mathsf{V}^\lambda)^2}{d!}
		\end{equation}
	 
	 \noindent
	implements the normalized character of the regular representation.
	More precisely, if $C \in \mathcal{Z}(d)$ is a central element whose eigenvalue 
	in $\mathsf{V}^\lambda$ is $\omega_C(\lambda)$, then by the Fourier isomorphism
	
		\begin{equation}
			\C\group{S}(d) \longrightarrow \bigoplus_{\lambda \vdash d} \End \mathrm{V}^\lambda
		\end{equation}
		
	\noindent
	we have that the normalized character of $C$ in the regular representation is given by
		
		\begin{equation}
			\langle \omega_C \rangle = \sum_{\lambda \vdash d} \omega_C(\lambda)
			\frac{(\dim \mathsf{V}^\lambda)^2}{d!},
		\end{equation}
		
	\noindent
	the expected value of $\omega_C(\cdot)$ under Plancherel measure.
	An immediate consequence of Theorem \ref{thm:StringForm} is that the first $N$
	coupling coefficients of $I_N^{(m)}$ are symmetric polynomials in $mN$
	variables whose coefficients are Plancherel expectations.
	
		\begin{thm}
		\label{thm:PlancherelExpectations}
			For any $1 \leq d \leq N,$ we have
			
				\begin{equation*}
				\begin{split}
					I_N^{(0)d} &= N^{2d} \langle \Omega_{\frac{1}{N}}^{-1} \rangle, \\
					I_N^{(1)d} &= N^d \sum_{\alpha \vdash d} p_\alpha(a_1,\dots,a_N) \langle \omega_\alpha \Omega_{\frac{1}{N}}^{-1} \rangle, \\
					I_N^{(2)d}	&= \sum_{\alpha,\beta \vdash d} p_\alpha(a_1,\dots,a_N) p_\beta(b_1,\dots,b_N)
				\langle \omega_\alpha \Omega_{\frac{1}{N}}^{-1} \omega_\beta \rangle.
				\end{split}
			\end{equation*}
		\end{thm}
	
	We refer to the the parameter range $1 \leq d \leq N$
	as the stable range. In the unstable range, where
	$d > N,$ the string form of the coupling coefficients involves an incomplete sum over Young diagrams
	$\lambda \vdash d$ against the Plancherel weight due to the restriction $\ell(\lambda) \leq N$, which is needed to
	ensure that the product
	
		\begin{equation}
			\Omega_{\frac{1}{N}}^{-1}(\lambda) = \prod_{\Box \in \lambda} 
			\frac{1}{1+\frac{c(\Box)}{N}}
		\end{equation}
		
	\noindent
	does not contain $\frac{1}{0}$ among its factors.
				
	Observe that, for any diagram $\lambda \vdash d$, the product 
	
		\begin{equation}
			\Omega_\hbar(\lambda) = \prod_{\Box \in \lambda} (1+\hbar c(\Box))
		\end{equation}
		
	\noindent
	is a degree $d$ polynomial in $\hbar$ whose roots are the reciprocals of the
	nonzero contents of the diagram $\lambda$. That is, we have
	
		\begin{equation}
			\Omega_\hbar(\lambda) = \sum_{r=0}^d \hbar^r e_r(\lambda),
		\end{equation}
		
	\noindent
	where $e_r(\lambda)$ denotes the complete symmetric function of degree $r$ 
	evaluated on the content alphabet of $\lambda$. It follows that 
	$\Omega_\hbar^{-1}(\lambda)$ is a holomorphic function of $\hbar$ on the disc 
	$|\hbar|<\frac{1}{d-1}$ whose Macluarin expansion is
	
		\begin{equation}
			\Omega_\hbar^{-1}(\lambda) = \sum_{r=0}^\infty (-\hbar)^r f_r(\lambda),
		\end{equation}
		
	\noindent
	where $f_r(\lambda)$ is the complete homogeneous symmetric function of degree $r$
	evaluated on the content alphabet of $\lambda$. We thus have the following 
	absolutely convergent $1/N$ expansions of the stable coupling coefficients of
	$I_N$ in terms of Plancherel averages.
			
		\begin{thm}
		\label{thm:StatisticalExpansion}
		For any $1 \leq d \leq N$, we have
		the absolutely convergent series expansions
		
			\begin{equation*}
			\begin{split}	
				I_N^{(0)d} &= N^{2d} \sum_{r=0}^\infty \left( -\frac{1}{N} \right)^r \langle f_r  \rangle, \\
				I_N^{(1)d} &= N^d \sum_{\alpha \vdash d} p_\alpha(a_1,\dots,a_N)
					\sum_{r=0}^\infty \left( -\frac{1}{N} \right)^r \langle \omega_\alpha f_r  \rangle, \\	
				I_N^{(2)d} &= \sum_{\alpha,\beta \vdash d} p_\alpha(a_1,\dots,a_N)p_\beta(b_1,\dots,b_N)
					\sum_{r=0}^\infty \left( -\frac{1}{N} \right)^r \langle \omega_\alpha f_r \omega_\beta \rangle.
			\end{split}
			\end{equation*}
		
		\end{thm}

	We remark that although the series 
	
		\begin{equation}
			 \sum_{r=0}^\infty \left( -\frac{1}{N} \right)^r \langle f_r  \rangle
		\end{equation}
		
	\noindent
	appears to be an alternating series, this is in fact not the case,
	because the terms corresponding to odd values of $r$ are equal to zero. Indeed,
	for any Young diagram $\lambda$, the content alphabet of the conjugate diagram $\lambda^*$
	is the negative of the content alphabet of $\lambda$, so that 
	
		\begin{equation}
			f_r(\lambda^*) = (-1)^r f_r(\lambda)
		\end{equation}
		
	\noindent
	by the homogeneity of $f_r$. We will see shortly that a similar statement holds in general for the series 
	
		\begin{equation}
			\sum_{r=0}^\infty \left( -\frac{1}{N} \right)^r \langle \omega_\alpha f_r \omega_\beta \rangle,
		\end{equation}
		
	\noindent 
	whose nonzero terms are either all positive or all negative.
		
	\subsection{Graphical expansion}
	Theorem \ref{thm:StatisticalExpansion} may be parlayed into a graphical expansion 
	in the manner of Feynman diagrams. This is done by inverting the Fourier transform on $\group{S}(d)$.	
	
	The unknown central element in $\C\group{S}(d)$ whose normalized trace in the regular representation is 
	the expectation $\langle \omega_\alpha f_r \omega_\beta \rangle$ must be of the form $C_\alpha F_r C_\beta,$
	where $F_r \in \mathcal{Z}(d)$ acts in $\mathsf{V}^\lambda$ as multiplication by $f_r(\lambda).$
	The central element $F_r$ may be expressed in terms of the Jucys-Murphy elements of the 
	group algebra $\C\group{S}(d)$,

		\begin{equation}
		\label{eqn:JMelements}
			X_t = \sum_{s < t} (s\ t), \quad 1 \leq t \leq d.
		\end{equation}
		
	\noindent
	Here $(s\ t) \in \group{S}(d)$ denotes the transposition interchanging the points $s<t$ in $\{1,\dots,d\}$.
	These special elements play a fundamental role in the representation theory of $\C\group{S}(d)$, see \cite{DG,OV}.
    	While the transposition sums \eqref{eqn:JMelements} commute with one another, they are not themselves central.
    	However, for any symmetric polynomial $p$, the group algebra element $p(X_1,\dots,X_d)$ lies in the center 
    	$\mathcal{Z}(d)$, and acts in $\mathsf{V}^\lambda$ as multiplication by the scalar $p(\lambda)$ obtained by
    	evaluating $p$ on the content alphabet of $\lambda$.
	We conclude that the Plancherel expectation $\langle \omega_\alpha f_r \omega_\beta \rangle$ is the normalized
	character of the central element 
		
		\begin{equation}
		\label{eqn:JMproduct}
			C_\alpha f_r(X_1,\dots,X_d) C_\beta
		\end{equation}
		
	\noindent
	acting the regular representation of $\C\group{S}(d)$, whence 
	$\langle \omega_\alpha f_r \omega_\beta \rangle$ is 
	the coefficient of the identity $\iota \in \group{S}(d)$ when the product
	\eqref{eqn:JMproduct} is expressed as a linear combination of conjugacy classes.
	
		\begin{figure}
			\includegraphics{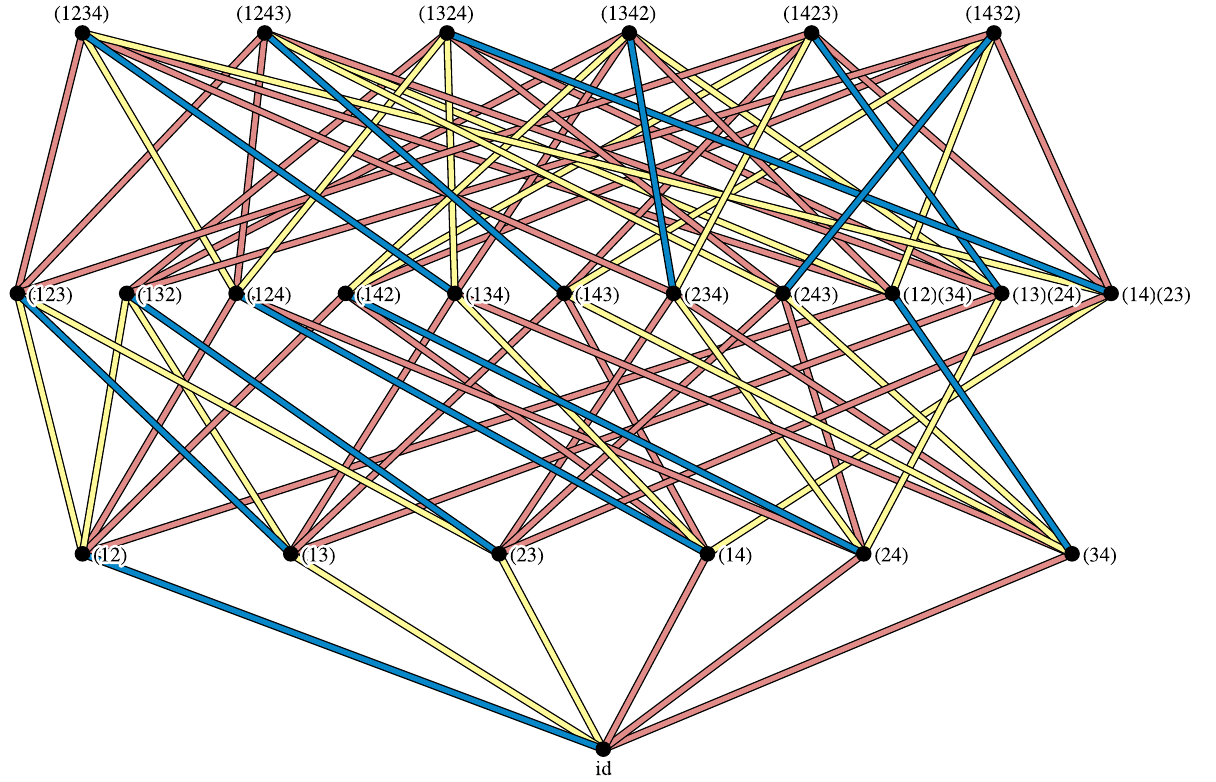}
			\caption{\label{fig:Cayley} Edge labeled Cayley graph of $\group{S}(4).$ Figure by M. LaCroix.}
		\end{figure}
		
	In order to visualize the coefficient of $\iota$ in \eqref{eqn:JMproduct},
	we consider the Cayley graph of $\group{S}(d)$, as generated by the conjugacy class of transpositions.
	Let us mark each edge of the Cayley graph corresponding to the transposition $(s\ t)$ with $t,$ 
	the larger of the two numbers
	interchanged. Figure \ref{fig:Cayley} shows the case $d=4$, with $2$-edges in blue, $3$-edges in yellow, and 
	$4$-edges in red. Let us call a walk on the Cayley graph \emph{monotone} if the labels of the edges it traverses form
	a weakly increasing sequence. Thus, monotone walks are virtual histories of the evolution of
	a particle on $\group{S}(d)$ which learns from experience --- once it travels along an edge of 
	value $t,$ it refuses to 
	traverse edges of lesser value.	Let $\vec{W}^r(\alpha,\beta)$ denote the number of monotone $r$-step 
	walks on the Cayley graph which 
	begin at a permutation of cycle type $\alpha$, and end at a permutation of cycle type $\beta$,
	i.e. the total number of monotone walks of length $r$ from the conjugacy class $C_\alpha$ to
	the conjugacy class $C_\beta$.
	Then, by definition of the complete symmetric polynomials and the Jucys-Murphy elements, 
	we see that 
		
		\begin{equation}
			\langle \omega_\alpha f_r \omega_\beta \rangle = \vec{W}^r(\alpha,\beta).
		\end{equation}
		
	\noindent
	Thus, Theorem \ref{thm:StatisticalExpansion} gives the following graphical expansion.
	
		\begin{thm}
		\label{thm:GraphicalExpansion}
		For any $1 \leq d \leq N$, we have
		the absolutely convergent series expansions
		
			\begin{equation*}
			\begin{split}	
				I_N^{(0)d} &= N^{2d} \sum_{r=0}^\infty \left( -\frac{1}{N} \right)^r \vec{W}^r(1^d,1^d), \\
				I_N^{(1)d} &= N^d \sum_{\alpha \vdash d} p_\alpha(a_1,\dots,a_N)
					\sum_{r=0}^\infty \left( -\frac{1}{N} \right)^r \vec{W}^r(\alpha,1^d), \\	
				I_N^{(2)d} &= \sum_{\alpha,\beta \vdash d} p_\alpha(a_1,\dots,a_N)p_\beta(b_1,\dots,b_N)
					\sum_{r=0}^\infty \left( -\frac{1}{N} \right)^r \vec{W}^r(\alpha,\beta).
			\end{split}
			\end{equation*}
		
		\end{thm}

		Theorem \ref{thm:GraphicalExpansion} says that the 
		stable string coefficients $I_N^d$ are generating functions 
		for monotone walks on the symmetric group with endpoints in given 
		conjugacy classes: 
		with the case $m=0$ corresponds to loops based at a given point of $\group{S}(d)$, 
		the case $m=1$ corresponding to anchored walks issuing from the 
		identity and ending in a given class, and the
		case $m=2$ allowing arbitrary conjugacy classes as boundary conditions.		
		As a consistency check on the graphical expansion, observe that since 
	
			\begin{equation}
				\lim_{N \to \infty} \Omega_{\frac{1}{N}}(\lambda)=1
			\end{equation}
		
		\noindent
		for each fixed $\lambda \vdash d$, we have
	
			\begin{equation}
				 \lim_{N \to \infty} \langle \omega_\alpha
				\Omega_{\frac{1}{N}}^{-1}\omega_\beta\rangle = \delta_{\alpha\beta} |C_\alpha|
			\end{equation}
		
		\noindent
		for each fixed $\alpha,\beta \vdash d$, where the second equality is the orthogonality of irreducible characters.
		This comports with the combinatorially obvious enumeration of $0$-step walks from 
		$C_\alpha$ to $C_\beta$,
		
			\begin{equation}
				 \vec{W}^0(\alpha,\beta) = \delta_{\alpha\beta} |C_\alpha|.
			\end{equation}
		
		\subsection{Topological expansion}
		The fact that monotone walks on symmetric groups play the role of Feynman diagrams for integration against the 
		Haar measure on $\group{U}(N)$ was discovered in \cite{Novak:Banach}, and further developed
		in \cite{MN1,MN2}. It was subsequently understood that these trajectories admit a natural topological interpretation involving branched covers
		of the sphere \cite{GGN1,GGN2}.
								
		To see this, let us recall the relationship between the number $W^r(\alpha,\beta)$ of not necessarily monotone $r$-step walks on $\group{S}(d)$ 
		between conjugacy classes $C_\alpha$ and $C_\beta$ and maps to the sphere. Applied in reverse, Hurwitz's classical monodromy construction \cite{Hurwitz1}
		interprets $\frac{1}{d!}W^r(\alpha,\beta)$ as a weighted count of isomorphism classes of pairs $(\mathbf{X},f)$ consisting of a compact Riemann surface
		$\mathbf{X}$ together with a degree $d$ holomorphic map $f \colon \mathbf{X} \to \mathbf{P}^1$ to 
		the Riemann sphere with ramification profiles $\alpha,\beta$ over $\infty,0 \in \mathbf{P}^1$, and the simplest non-trivial
		branching over the $r$th roots of unity on the sphere. The normalized counts $\frac{1}{d!}W^r(\alpha,\beta)$
		were called the disconnected double Hurwitz numbers in \cite{Okounkov:MRL}; for our purposes, it 
		is more convenient to assign this name to the raw count $W^r(\alpha,\beta)$.				
		By the Riemann-Hurwitz formula, the disconnected double Hurwitz number $W^r(\alpha,\beta)$
		is zero unless 
				
			\begin{equation}
			\label{eqn:RiemannHurwitz}
				r = 2g-2+\ell(\alpha)+\ell(\beta),
			\end{equation}
			
		\noindent
		where $g=g(\mathbf{X})$ is the genus of $\mathbf{X}$. Note that if 
				
			\begin{equation}
				\mathbf{X} = \mathbf{X}_1 \sqcup \dots \sqcup \mathbf{X}_c
			\end{equation}
			
		\noindent
		is a disjoint union of $c$ connected components, then
		
			\begin{equation}
				g(\mathbf{X}) = \left(g(\mathbf{X}_1)-1\right) + \dots +  \left(g(\mathbf{X}_c)-1\right) + 1
			\end{equation}
		
		\noindent
		by additivity of the Euler characteristic, so that the genus of a disconnected 
		surface may be negative but is subject to the lower bound
		
			\begin{equation}
			\label{eqn:LowerGenusBound}
				g(\mathbf{X})  \geq -c + 1.
			\end{equation}
			
		\noindent
		In view of \eqref{eqn:RiemannHurwitz}, 
		we may re-index the disconnected double Hurwitz numbers by genus, setting
						
			\begin{equation}
			\label{eqn:DisconnectedHurwitzNumbers}
				H_g^\bullet(\alpha,\beta) := W^{2g-2+\ell(\alpha) + \ell(\beta)}(\alpha,\beta),
			\end{equation}
			
		\noindent
		where the bullet indicates disconnected. 
		The case $\beta=1^d$ corresponds to covers unramified over $0 \in \mathbf{P}^1$,
		and it is customary to write
		
			\begin{equation}
				H_g^\bullet(\alpha) := H_g^\bullet(\alpha,1^d) 			
			\end{equation}
			
		\noindent
		and call these the disconnected single Hurwitz numbers; these were the numbers
		originally studied by Hurwitz \cite{Hurwitz1,Hurwitz2}. We may further define 
		the disconnected simple Hurwitz numbers by 
		
			\begin{equation}
				H_g^{\bullet d}:= H_g^\bullet(1^d).
			\end{equation}
			
		\noindent
		A good reference on the combinatorial features of simple Hurwitz numbers is \cite{DYZ}.
								
		By analogy with the above, the monotone walk counts $\vec{W}^r(\alpha,\beta)$ 
		and $\vec{W}^r(\alpha,1^d)$ 
		were termed the disconnected monotone double and single Hurwitz numbers in \cite{GGN1,GGN2,GGN3}. 
		Since these nonnegative integers are, by construction, smaller than their classical
		non-monotone counterparts, the Riemann-Hurwitz formula imposes the 
		same vanishing condition on monotone Hurwitz numbers as it does on classical 
		Hurwitz numbers, and accordingly we may re-index monotone Hurwitz numbers by genus, 
		setting 
		
			\begin{equation}
			\label{eqn:DisconnectedMonotoneHurwitzNumbers}
				\mon_g^\bullet(\alpha,\beta) := \vec{W}^{2g-2+\ell(\alpha)+\ell(\beta)}(\alpha,\beta)
				\quad\text{ and }\quad \mon_g^\bullet(\alpha) := \mon_g^\bullet(\alpha,1^d).
			\end{equation}
							
		\noindent
		We further define the disconnected monotone simple Hurwitz numbers by
		
			\begin{equation}
				\mon_g^{\bullet d}:= \mon_g^\bullet(1^d).
			\end{equation}

		Topologically, monotone Hurwitz numbers may be viewed as corresponding to 
		a signed enumeration of the the same class of covers
		counted by classical Hurwitz numbers; despite being signed, this count is always
		nonnegative. The classical and monotone constructions have been unified within a more
		general theory of weighted Hurwitz numbers \cite{ACEH}, but remain 
		the two most important and useful instances of this generalization. 		
		In particular, we may now interpret the graphical expansions given in Theorem \ref{thm:GraphicalExpansion} as topological expansions.

			\begin{thm}
			\label{thm:TopologicalExpansion}
			For any integers $1 \leq d \leq N$, we have
			
				\begin{equation*}
					I_N^d = \sum_{g=-d+1}^\infty N^{2-2g} I_{Ng}^d,
				\end{equation*}
				
			\noindent
			where the series is $\|\cdot\|$-absolutely convergent, and the 
			coefficients $I_{Ng}^d=I_{Ng}^{(m)d}$ are homogeneous degree
			$d$ polynomials in $mN$ variables given by
		
			\begin{equation*}
			\begin{split}
				I_N^{(0)d} &=  \mon_g^{\bullet d}, \\	
				I_N^{(1)d} &=  \sum_{\alpha \vdash d} \frac{p_\alpha(a_1,\dots,a_N)}{N^{\ell(\alpha)}} (-1)^{\ell(\alpha)+d}
					 \mon_g^\bullet(\alpha), \\	
				I_N^{(2)d} &=\sum_{\alpha,\beta \vdash d} \frac{p_\alpha(a_1,\dots,a_N)}{N^{\ell(\alpha)}}
				\frac{p_\beta(b_1,\dots,b_N)}{N^{\ell(\beta)}} (-1)^{\ell(\alpha)+\ell(\beta)}
					 \mon_g^\bullet(\alpha,\beta).
			\end{split}
			\end{equation*}
			\end{thm}

	\section{Infinite $N$}
	\label{sec:Infinite}
	From Section \ref{sec:Finite}, we know that for sufficiently strong coupling the 
	large $N$ behavior of $I_N$ is captured by its critical coupling coefficients
	
		\begin{equation}
			I_N^1,\dots,I_N^{\lfloor tN^2\rfloor}.
		\end{equation}
	
	\noindent
	However, only the first $N$ of these, the stable coupling coefficients
	
		\begin{equation}
			I_N^1,\dots,I_N^N,
		\end{equation}
		
	\noindent
	admit convergent topological expansions. We will bridge the gap between 
	the stable and critical ranges analytically, as $N \to \infty$, in Section \ref{sec:Large}. 
	In this section we do so formally, by setting $N=\infty$.
	
	\subsection{Stable partition functions}
	From the matrix integral perspective, setting $N=\infty$ in our problem
	corresponds to replacing the integrals \eqref{eqn:BGW} and 
	\eqref{eqn:HCIZ} with 
				
		 \begin{equation}
		    \label{eqn:StableBGW}
			I^{(1)} = \int_{\group{U}} e^{\sqrt{z}\hbar^{-1} \Tr (AU + BU^{-1})} \mathrm{d}U,
		\end{equation}
		
	\noindent
	and

		\begin{equation}
		\label{eqn:StableHCIZ}
			I^{(2)} = \int_{\group{U}} e^{z\hbar^{-1} \Tr AUBU^{-1}} \mathrm{d}U,
		\end{equation}

	\noindent
	where the integration is over the stable unitary group \cite{Bott,SV}
	
		\begin{equation}
		\label{eqn:StableUnitaryGroup}
			U = \varinjlim \group{U}(N)
		\end{equation}
		
	\noindent
	with $\hbar$ an infinitely small parameter and $A,B$ infinitely large matrices.
	However, as $\group{U}$ is not locally compact and does not support a Haar measure,
	and these are ill-defined functional integrals. 
	
	Another approach to stability is to set $N = \infty$ in the coupling expansion
	of $I_N$, which means that we consider power series 
	
		\begin{equation}
			I^{(m)} = 1+ \sum_{d=1}^\infty \frac{z^d}{d!} I^{(m)d}, \quad m=0,1,2
		\end{equation}
		
	\noindent
	in a formal coupling constant $z$ whose coefficients $I^{(m)d}$ are themselves
	formal series in an indeterminate $\hbar$ representing the parameter $1/N$ at $N = \infty$. 
	More precisely, the coupling coefficients $I^{(m)d}$ may be constructed as statistical expansions 
	which are $N=\infty$ versions of the statistical expansions of $I_N^{(m)d}$ 
	given by Theorem \ref{thm:StatisticalExpansion},

		\begin{equation}
		\label{eqn:StableCouplingStatistical}
		\begin{split}
			I^{(0)d} &= \hbar^{-2d} \sum_{r=0}^\infty (-\hbar)^r \langle f_r \rangle\\
			I^{(1)d} &= \hbar^{-d}\sum_{\alpha \vdash d} p_\alpha(A) \hbar^{-d}
				\sum_{r=0}^\infty (-\hbar)^r \langle \omega_\alpha f_r \rangle \\
			I^{(2)d} &=  \sum_{\alpha,\beta \vdash d} p_\alpha(A) p_\beta(B)
				\sum_{r=0}^\infty (-\hbar)^r \langle \omega_\alpha f_r \omega_\beta \rangle,
		\end{split}
		\end{equation}

	\noindent
	or as graphical expansions which are $N=\infty$ versions of the graphical expansions of $I_N^{(m)d}$ 
	given by Theorem \ref{thm:GraphicalExpansion},
			
		\begin{equation}
		\label{eqn:StableCouplingCombinatorial}
		\begin{split}
			I^{(0)d} &= \hbar^{-2d} \sum_{r=0}^\infty (-\hbar)^r \vec{W}^r(1^d,1^d)\\
			I^{(1)d} &= \hbar^{-d}\sum_{\alpha \vdash d} p_\alpha(A) \hbar^{-d}
				\sum_{r=0}^\infty (-\hbar)^r \vec{W}^r(\alpha,1^d) \\
			I^{(2)d} &=  \sum_{\alpha,\beta \vdash d} p_\alpha(A) p_\beta(B)
				\sum_{r=0}^\infty (-\hbar)^r \vec{W}^r(\alpha,\beta),
		\end{split}
		\end{equation}
		
	\noindent
	or as topological expansions which are $N=\infty$ versions of the topological expansions of $I_N^{(m)d}$ 
	given by Theorem \ref{thm:GraphicalExpansion},

		\begin{equation}
		\label{eqn:StableCouplingTopological}
		\begin{split}
			I^{(0)d} &= \sum_{g=-d+1} \hbar^{2g-2} \mon_g^{\bullet d} \\
			I^{(1)d} &=  \sum_{\alpha \vdash d} \frac{p_\alpha(A)}{\hbar^{-\ell(\alpha)}} (-1)^{\ell(\alpha)+d}
				\sum_{g=-d+1}^\infty \hbar^{2g-2} \mon_g^\bullet(\alpha) \\
			I^{(2)d} &= \sum_{\alpha \vdash d} \frac{p_\alpha(A)}{\hbar^{-\ell(\alpha)}} \frac{p_\alpha(B)}{\hbar^{-\ell(B)}} (-1)^{\ell(\alpha)+\ell(\beta)}
				\sum_{g=-d+1}^\infty \hbar^{2g-2} \mon_g^\bullet(\alpha,\beta).
		\end{split}
		\end{equation}
		
	\noindent
	Here 
	
		\begin{equation}
			p_\alpha(A) = \prod_{i=1}^{\ell(\alpha)} \sum_{j=1}^\infty a_j^{\alpha_i} \quad\text{ and }\quad
			p_\beta(B) = \prod_{i=1}^{\ell(\beta)} \sum_{j=1}^\infty b_j^{\beta_i} 
		\end{equation}
		
	\noindent
	are the Newton power sum symmetric functions in two
	countably infinite alphabets $A=\{a_1,a_2,\dots\}$ and $B=\{b_1,b_2,\dots\}$ of formal 
	variables playing the role of the external matrix fields in the integrals
	\eqref{eqn:BGW} and \eqref{eqn:HCIZ}. Observe that the signs which 
	appear in these expressions could be eliminated simply by replacing
	$\hbar$ with $-\hbar$, but we refrain from doing this in order to maintain notational 
	consistency with Section \ref{sec:Finite}, and because we wish to emphasize that 
	$\hbar$ is a formal replacement for $1/N$, not for $-1/N$. We view the generating functions
	$I^{(1)}$ and $I^{(2)}$ as meaningful algebraic versions of the ill-defined functional integrals 
	\eqref{eqn:StableBGW} and \eqref{eqn:StableHCIZ}.

	\subsection{Stable free energy}
	A subtle but important feature of the total generating function $I$ for disconnected
	monotone Hurwitz theory is that the variable $\hbar$ is an 
	ordinary marker for Euler characteristic, whereas in classical 
	Hurwitz theory one uses an exponential marker for this statistic \cite{GJV,Okounkov:MRL}.
	In a sense, this particularity explains the inevitability of monotone Hurwitz theory: the large $N$ expansion \eqref{eqn:LargeNExpansion}
	presents as an ordinary generating function in $1/N$, not an exponential one,
	so that Hurwitz theory must be desymmetrized in order to match it with the 't Hooft expansion.
	Crucially, we have the following theorem from \cite{GGN1,GGN2}, which says that even though
	$I$ is a mixed ordinary/exponential generating function, taking the 
	logarithm does indeed extract connected information.

		\begin{thm}
		\label{thm:ExponentialFormula}
			We have
						
				\begin{equation*}
					\log I^{(m)} = F^{(m)},
				\end{equation*}

			\noindent
			where 
			
				\begin{align*}
					F^{(0)} &= \sum_{d=1}^\infty \frac{z^d}{d!} \sum_{g=0}^\infty \hbar^{2g-2} \mon_g^d\\
					F^{(1)} &= \sum_{d=1}^\infty \frac{z^d}{d!} 
						\sum_{\alpha \vdash d} \frac{p_\alpha(A)}{\hbar^{-\ell(\alpha)}} (-1)^{\ell(\alpha)+d}
							\sum_{g=0}^\infty \hbar^{2g-2} \mon_g(\alpha) \\
					F^{(2)} &=\sum_{d=1}^\infty \frac{z^d}{d!} 
					\sum_{\alpha \vdash d} \frac{p_\alpha(A)}{\hbar^{-\ell(\alpha)}} \frac{p_\alpha(B)}{\hbar^{-\ell(B)}} (-1)^{\ell(\alpha)+\ell(\beta)}
					\sum_{g=0}^\infty \hbar^{2g-2} \mon_g(\alpha,\beta).
				\end{align*}
				
			\noindent
			are total generating functions for the connected monotone simple, single, and double
			Hurwitz numbers in all degrees and genera.
		\end{thm}
		
	An immediate corollary of Theorem \ref{thm:ExponentialFormula} is that a formal version
	of Conjecture \ref{conj:Main} holds at $N=\infty$, i.e. the stable free energy $F=\log I$
	admits a topological expansion.
	
		\begin{thm}
		\label{thm:StableTopologicalExpansion}
			We have 
			
				\begin{equation*}
					\log I = \sum_{g=0}^\infty \hbar^{2g-2} F_g,
				\end{equation*}
				
			\noindent
			where 
			
				\begin{align*}
					F_g^{(0)} &= \sum_{d=1}^\infty \frac{z^d}{d!} \mon_g^d\\
					F_g^{(1)} &= \sum_{d=1}^\infty \frac{z^d}{d!} 
						\sum_{\alpha \vdash d} \frac{p_\alpha(A)}{\hbar^{-\ell(\alpha)}} (-1)^{\ell(\alpha)+d}\mon_g(\alpha) \\
					F_g^{(2)} &=\sum_{d=1}^\infty \frac{z^d}{d!} 
					\sum_{\alpha \vdash d} \frac{p_\alpha(A)}{\hbar^{-\ell(\alpha)}} \frac{p_\alpha(B)}{\hbar^{-\ell(B)}} (-1)^{\ell(\alpha)+\ell(\beta)}
					\mon_g(\alpha,\beta).
				\end{align*}
				
			\noindent
			are generating functions for the connected monotone simple, single, and double
			Hurwitz numbers in specified genus $g \geq 0$.		
		\end{thm}
		
	\noindent
	We emphasize that even this formal algebraic result has not previously appeared in the literature in this entirety.
	The case $m=2$ is implicit in \cite{GGN3}. The total generating function for monotone single Hurwitz numbers
	was studied in \cite{GGN1,GGN2} where low-genus explicit formulas and rational parameterizations were found, 
	without the understanding that what was being analyzed was the stable free energy
	of the BGW model.

	\subsection{Topological expansion as topological factorization}
	Our approach to Conjecture \ref{conj:Main} is based on converting the stable topological expansion of $F=\log I$ given by Theorem \ref{thm:StableTopologicalExpansion},
	which is purely algebraic, into an analytically meaninfgul $N \to \infty$ asymptotic expansion of
	$F_N=\log I_N$ which holds at sufficiently strong coupling. However, this cannot be done 
	directly at the free energy level, since we do not know whether or not the stable nonvanishing hypothesis holds, 
	and a priori it is not clear that we can view $F_N$ as an element of $\O_N(\delta)$ for $\delta>0$ an absolute
	constant. We therefore have to work with $I_N$ itself, and consequently it is beneficial to 
	reformulate topological expansion of the stable free energy $F$ as topological factorization of the stable partition function $I$.
	More precisely, Theorem \ref{thm:StableTopologicalExpansion} is equivalent to the 
	statement that, for each $k \in \N_0$, we have the factorization
	
		\begin{equation}
		\label{eqn:TopologicalFactorization}
			I = E_{\overline{k}} E_{\underline{k+1}},
		\end{equation} 
		
	\noindent
	where 
		
		\begin{equation}
			E_{\overline{k}} = e^{\sum_{g=0}^k \hbar^{2g-2}F_g} = 1+ \sum_{d=1}^\infty \frac{z^d}{d!} E_{\overline{k}}^d
		\end{equation}
		
	\noindent
	and 
	
		\begin{equation}
			E_{\underline{k+1}} = e^{\sum_{g=k+1}^\infty \hbar^{2g-2}F_g} = 1+ \sum_{d=1}^\infty \frac{z^d}{d!} E_{\underline{k+1}}^d
		\end{equation}
		
	\noindent
	are generating functions for disconnected covers built from connected components of genus at most $k$
	and at least $k+1$, respectively. For example, in the notationally simplest case $m=0$, we have
	
		\begin{equation}
			E_{\overline{k}}^d = \sum_{g=-d+1}^\infty \hbar^{2g-2} \mon_{g\overline{k}}^{\bullet d}
		\end{equation}
		
	\noindent
	with $\mon_{g\overline{k}}^{\bullet d}$ the number of degree $d$ 
	disconnected simple covers of genus $g$ all of whose connected components
	have genus at most $k$, while 
	
		\begin{equation}
			E_{\underline{k+1}}^d = \sum_{g=k+1}^\infty \hbar^{2g-2} \mon_{g\underline{k+1}}^{\bullet d}
		\end{equation}
		
	\noindent
	is a generating function for the number $\mon_{g\underline{k+1}}^{\bullet d}$ of 
	disconnected simple covers of genus $g$ all of whose connected components
	have genus at least $k+1$, and the sum starts at $g=k+1$ because any such 
	cover must itself have genus $k+1$. At the level of coupling coefficients,
	Theorem \ref{thm:StableTopologicalExpansion} is 
	equivalent to a binomial convolution formula
	describing the construction of an arbitrary degree $d$ cover as a shuffle of 
	two smaller-degree disconnected covers built from connected components
	of genus at most $k$ and at least $k+1$, respectively.
		
		\begin{thm}
		\label{thm:StableTopologicalConvolution}
		For any $d \in \N$ and $k \in \N_0$, we have
			
			\begin{equation*}
				I^d = \sum_{c=0}^d {d \choose c} E_{\overline{k}}^c E_{\underline{k+1}}^{d-c}.
			\end{equation*}
		\end{thm}

	\subsection{Topological factorization as topological concentration}
	Viewing the stable partition function $I$ as a formal replacement for an infinite-dimensional
	functional integral (e.g. \eqref{eqn:StableBGW} or \eqref{eqn:StableHCIZ}), 
	it is natural to further recast the topological factorization \eqref{eqn:TopologicalFactorization} 
	as a ``concentration inequality'' for the corresponding integral. More precisely, 
	let us define the $k$th order topological normalization of the stable partition function
	$I$ to be the series

		\begin{equation}
			\Phi_k = E_{\underline{k}}^{-1} I.
		\end{equation}
		
	\noindent
	In the case $m=1$, this is a formal replacement for the topological normalization of the stable 
	BGW integral,
	
		\begin{equation}
		    \label{eqn:TopologicallyNormalizedStableBGW}
			\Phi_k^{(1)} = e^{-\sum_{g=0}^k \hbar^{2g-2} F_g^{(1)}}\int_{\group{U}} e^{\sqrt{z}\hbar^{-1} \Tr (AU + BU^{-1})} \mathrm{d}U,
		\end{equation}
		
	\noindent
	and in the case $m=2$ it is a formal algebraic version of the topologically normalized
	HCIZ integral,
	
		\begin{equation}
		\label{eqn:StableHCIZ}
			\Phi_k^{(2)} = e^{-\sum_{g=0}^k \hbar^{2g-2} F_g^{(2)}}\int_{\group{U}} e^{z\hbar^{-1} \Tr AUBU^{-1}} \mathrm{d}U.
		\end{equation}

	\noindent
	The topological factorization of $I$ is then equivalent to the identity 
	
		\begin{equation}
			\Phi_k = E_{\underline{k+1}},
		\end{equation}
		
	\noindent
	which at the level of coupling coefficients is equivalent to the ``topological cancellation''
	identity
	
		\begin{equation}
			\label{eqn:StableTopologicalCancellation}
				\Phi_k^d = \sum_{c=0}^d {d \choose c} I^c (E_k^{-1})^{d-c} = E_{k+1}^d,
		\end{equation}
		
	\noindent 
	where 
	
		\begin{equation}
			E_k^{-1} = e^{-\sum_{g=0}^k \hbar^{2g-2}F_g} = 1 + \sum_{d=1}^\infty \frac{z^d}{d!} (E_k^{-1})^d
		\end{equation}
		
	\noindent
	the coupling expansion of the topological normalization factor $E_k^{-1}$.
	Since $E_{\underline{k+1}}$ is a generating function for covers built from
	connected components of genus at least $k+1$, this implies the algebraic 
	concentration inequality
		
		\begin{equation}
		\label{eqn:StableTopologicalConcentration}
			\Phi_k-1 = O(\hbar^{2k}),
		\end{equation}
	
	\noindent
	We will see in Section \ref{sec:Large} that proving Conjecture \ref{conj:Main}
	can be reduced to establishing an $N \to \infty$ version of \eqref{eqn:StableTopologicalConcentration}
	which holds at sufficiently strong coupling.

	\subsection{Summability in fixed genus}
	The following theorem from \cite{GGN5} forms the bridge between the 
	stable world of formal power series and the analytic world of holomorphic functions.
	
		\begin{thm}
		\label{thm:SimpleSummability}
			For each $g \in \N_0$, the generating function $F_g^{(0)}$
			for connected monotone simple Hurwitz numbers of genus $g$ has
			radius of convergence equal to $2/27$.
		\end{thm}
		
	Theorem \ref{thm:SimpleSummability} is deduced in \cite{GGN5} from the results of \cite{GGN2}, which 
	imply that the series $F_g^{(0)}$ can be expressed in terms of the classical Gaussian hypergeometric function.
	Interestingly, the constant $2/27$ also appears in the asymptotic enumeration of finite groups \cite{Pyber}.
	For other appearances of $2/27$ see \cite{KT} and \cite{CMZagier}, the latter reference being 
	more directly (but still not transparently) related to the present context.
	
	Theorem \ref{thm:SimpleSummability} can be lifted to summability 
	of the analogous genus-specific generating functions for the connected 
	single and double monotone Hurwitz numbers, via the connected version 
	of Theorem \ref{thm:DisconnectedSorting}, also proved in \cite{GGN5}.
				
		\begin{thm}
		\label{thm:ConnectedSorting}
			For any $d \in \N$ and $g \in \N_0$, 
			we have nd the same relation holds in the connected case,
			
				\begin{equation*}
					\sum_{\alpha,\beta \vdash d} \mon_g(\alpha,\beta) 
					< 2^d \sum_{\alpha \vdash d} \mon_g(\alpha) < 4^d \mon_g^d.
				\end{equation*}

		\end{thm}
		
		\subsection{Large genus asymptotics}
		Theorem \ref{thm:SimpleSummability} corresponds to an asymptotic
		result for the connected monotone simple Hurwitz numbers in the 
		large degree limit: for each fixed $g \in \N_0$, we have
		
			\begin{equation}
				\frac{1}{d!} \mon_g^d \sim t_g(d) \left(\frac{27}{2} \right)^d, \quad d \to \infty,
			\end{equation}
			
		\noindent
		with $t_g(d)$ a genus-dependent factor of sub-exponential growth in $d$.
		Below, we will also need information on the large genus asymptotics 
		of disconnected monotone simple Hurwitz numbers in fixed degree.
		These asymptotics are as follows.
		
			\begin{thm}
			\label{thm:LargeGenusAsymptotics}
				For any $d \in \N$, we have
				
					\begin{equation*}
						\mon_g^{\bullet d} \sim \frac{2(d-1)^{3d-3}}{(d-1)!d!} (d-1)^{2g}, \quad g \to \infty.					
					\end{equation*}
			\end{thm}
			
			\begin{proof}
				Let $d \in \N$ be arbitrary but fixed.
				Since $\mon_g^{\bullet d} = \vec{W}^{2g-2+2d}(1^d,1^d)$, the $g \to \infty$ 
				asymptotic behavior of the monotone Hurwitz number $\mon_g^{\bullet d}$ is equivalent to the $r \to \infty$ asymptotic behavior
				of the monotone loop counts $\vec{W}^r(1^d,1^d)$. 
				
				From Section \ref{sec:Finite}, we know that the ordinary generating function for monotone loops $\vec{W}^r(1^d,1^d)$ based 
				at a given point of the symmetric group is a rational function of $\hbar$,
		
					\begin{equation}
						\sum_{r=0}^\infty \hbar^r \vec{W}^r(1^d,1^d) = \sum_{\lambda \vdash d} \frac{(\dim \mathsf{V}^\lambda)^2}{d!}
						\prod_{\Box \in \lambda} \frac{1}{1-\hbar c(\Box)}.
					\end{equation}
			
				\noindent
				It follows that, as a function of $r$, the loop count $\vec{W}^r(1^d,1^d)$ is a linear combination of the exponential
				functions
				
					\begin{equation}
						1^r, 2^r, \dots, (d-1)^r
					\end{equation}
					
				\noindent
				whose coefficients are polynomials in $r$. In other words,
				$\vec{W}^r(1^d,1^d)$ is a quasipolynomial function of $r$.
				The dominant term of this quasipolynomial comes from those terms in the Plancherel average
				which have one of the factors $(1-(d-1)\hbar)$ or $(1+(d-1)\hbar)$ in their denominator,
				and the only two terms which meet this condition correspond
				to the trivial representation,
				
					\begin{equation}
						\frac{1}{d!}\frac{1}{(1-\hbar)(1-2\hbar) \dots (1-(d-1)\hbar)} = \frac{1}{d!} \sum_{r=0}^\infty \hbar^r 
						f_r(1,2,\dots,(d-1)),
					\end{equation}
					
				\noindent
				and the sign representation,
				
					\begin{equation}
						\frac{1}{d!}\frac{1}{(1+\hbar)(1+2\hbar) \dots (1+(d-1)\hbar)}
						= \frac{1}{d!} \sum_{r=0}^\infty \hbar^r 
						f_r(-1,-2,\dots,-(d-1)),
					\end{equation}
					
				\noindent
				so that by homogeneity we have
				
					\begin{equation}
						\vec{W}^r(1^d,1^d) \sim \frac{1^r + (-1)^r}{d!} f_r(1,\dots,d-1), \quad r \to \infty.
					\end{equation}
					
				Specializing the complete symmetric functions at 
				consecutive positive integers yields Stirling numbers of the second kind \cite{Macdonald},
				and in particular we have
				
					\begin{equation}
						f_r(1,2,\dots,d-1) = \stirling{d-1+r}{d-1}.
					\end{equation}
					
				\noindent
				The asymptotic behavior of Stirling numbers with large upper index and 
				fixed lower index is also known \cite{MW}, and in particular we have
				
					\begin{equation}
						\stirling{d-1+r}{d-1} \sim \frac{(d-1)^{d-1}}{(d-1)!} (d-1)^r, \quad r \to \infty.
					\end{equation}		
					
				\noindent 
				Taking $r=2g-2+2d$, we obtain
				
					\begin{equation}
						\vec{W}^r(1^d,1^d) \sim \frac{(1^r + (-1)^r )(d-1)^{d-1}}{(d-1)!d!} (d-1)^r, \quad r \to \infty
					\end{equation}
					
				\noindent
				and 
				
					\begin{equation}
						\mon_g^{\bullet d} \sim \frac{2(d-1)^{3d-3}}{(d-1)!d!} (d-1)^{2g}, \quad g \to \infty,
					\end{equation}
					
				\noindent
				as claimed.
			\end{proof}
						
		It is interesting to compare the asymptotics of monotone simple Hurwitz numbers catalogued here 
		with those of classical simple Hurwitz numbers, which are clearly catalogued in \cite{DYZ}.
		We end this section by remarking that the disconnnected version of the Theorem \ref{thm:ConnectedSorting}
		also holds; the proof is identical to the connected case.
		
			\begin{thm}
			\label{thm:DisconnectedSorting}
				For any $d \in \N$ and $g \in \Z$, we have
				
					\begin{equation*}
						\sum_{\alpha, \beta \vdash d} \mon_g^\bullet(\alpha,\beta) 
						< 2^d \sum_{\alpha \vdash d} \mon_g^\bullet(\alpha) 
						< 4^d \mon_g^{\bullet d}.
					\end{equation*}
			\end{thm}

\section{Large $N$}
\label{sec:Large}
In this Section, we combine the theorems of Sections \ref{sec:Finite} and  
\ref{sec:Infinite} with complex analytic methods to prove a 
generalization of Theorem \ref{thm:Main}, which in particular confirms Conjecture \ref{conj:Main}.
More precisely, we are concerned with the $N \to \infty$ asymptotics of the triarchy of entire functions
in $mN+1$ complex variables defined by

	\begin{equation}
		I_N^{(m)} = 1 + \sum_{d=1}^\infty \frac{z^d}{d!} I_N^{(m)d}, \quad m \in \{0,1,2\},\ N \in \N,
	\end{equation}
	
\noindent
whose coupling coefficients may be described in character form,

	\begin{equation}
	\label{eqn:TrinityCharacterForm}
		\begin{split}
			I_N^{(0)d} &= N^{2d} \sum_{\substack{\lambda \vdash d \\ \ell(\lambda) \leq N}} 
				\prod_{\Box \in \lambda} \frac{1}{h(\Box)^2(1+\frac{c(\Box)}{N})}, \\
			I_N^{(1)d} &=  N^d \sum_{\substack{\lambda \vdash d \\ \ell(\lambda) \leq N}} s_\lambda(a_1,\dots,a_N)
				\prod_{\Box \in \lambda} \frac{1}{h(\Box)(1+\frac{c(\Box)}{N})}, \\
			I_N^{(2)d} &= \sum_{\substack{\lambda \vdash d \\ \ell(\lambda) \leq N}} s_\lambda(a_1,\dots,a_N) s_\lambda(b_1,\dots,b_N)
				\prod_{\Box \in \lambda} \frac{1}{1+\frac{c(\Box)}{N}},
		\end{split}
	\end{equation}
	
\noindent
or equivalently in string form,

	\begin{equation}
	\label{eqn:TrinityStringForm}
		\begin{split}
			I_N^{(0)d} &= N^{2d} 
				\sum_{\substack{\lambda \vdash d \\ \ell(\lambda) \leq N}} \Omega_{\frac{1}{N}}^{-1}(\lambda) \frac{(\dim \mathsf{V}^\lambda)^2}{d!}, \\
			I_N^{(1)d} &= N^d \sum_{\alpha \vdash d} p_\alpha(a_1,\dots,a_N)
				\sum_{\substack{\lambda \vdash d \\ \ell(\lambda) \leq N}} 
				\omega_\alpha(\lambda)\Omega_{\frac{1}{N}}^{-1}(\lambda) \frac{(\dim \mathsf{V}^\lambda)^2}{d!}, \\
			I_N^{(2)d} &= \sum_{\alpha \vdash d} p_\alpha(a_1,\dots,a_N) p_\beta(b_1,\dots,b_N)
				\sum_{\substack{\lambda \vdash d \\ \ell(\lambda) \leq N}} 
				\omega_\alpha(\lambda)\Omega_{\frac{1}{N}}^{-1}(\lambda)\omega_\beta(\lambda) \frac{(\dim \mathsf{V}^\lambda)^2}{d!}.		
		\end{split}
	\end{equation}

\noindent
We continue the practice of omitting the superscript $m$
when making statements which hold uniformly in $m \in \{0,1,2\}$, when it is convenient
to do so.

Our goal is to obtain $N \to \infty$ asymptotics for 
$F_N^{(m)}=\log I_N^{(m)}$ on the closed origin-centered polydisc $\D_N(\varepsilon)$
of polyradius $(\varepsilon,1,\dots,1)$ in $\C^{mN+1}$,
with $\varepsilon > 0$ a sufficiently small absolute constant. At present, we do not 
know that this is a well-defined objective, as there may be no $\varepsilon >0$ such 
that $F_N^{(m)}$ is defined and analytic on $\D_N(\varepsilon)$, i.e. we do not
know that the stable nonvanishing hypothesis holds. However, Section \ref{sec:Infinite}
gives us explicit analytic targets, and also places
a limitation on how large $\varepsilon$ can be. Let $\gamma \in (0,\frac{1}{54})$ 
be fixed for the rest of the paper. Then, by the results of Section \ref{sec:Infinite}, for each 
$N \in \N$, the power series 

	\begin{equation}
		F_{Ng}^{(m)} = \sum_{d=1}^\infty \frac{z^d}{d!} F_{Ng}^{(m)d}, \quad g \in \N_0,\ m \in \{0,1,2\}
	\end{equation}

\noindent
whose coefficients are the polynomials in $mN$ variables defined by 
				
			\begin{equation}
			\label{eqn:GenusSpecificPolynomials}
			\begin{split}
				F_{Ng}^{(0)d} &=\mon_g^d \\
				F_{Ng}^{(1)d} &=  \sum_{\alpha \vdash d} \frac{p_\alpha(a_1,\dots,a_N)}{N^{\ell(\alpha)}} (-1)^{\ell(\alpha)+d}\mon_g(\alpha), \\
				F_{Ng}^{(2)d} &= \sum_{\alpha,\beta \vdash d} \frac{p_\alpha(a_1,\dots,a_N)}{N^{\ell(\alpha)}}\frac{p_\beta(b_1,\dots,b_N)}{N^{\ell(\beta)}}
					(-1)^{\ell(\alpha)+\ell(\beta)}\mon_g(\alpha,\beta)
			\end{split}
			\end{equation}
			
\noindent
are $\|\cdot\|_\gamma$-absolutely convergent, i.e. are members of $\O_N(\gamma)$, 
and 

	\begin{equation}
		\sup_{N \in \N} \|F_{Ng}\|_\gamma < \infty
	\end{equation}
	
\noindent
by construction.

In this Section, we prove the following theorem.

	\begin{thm}
	\label{thm:MainGeneralized}
		There exists $\varepsilon \in (0,\gamma)$ such that $I_N^{(m)}$ is non-vanishing
		on $\D_N(\varepsilon)$ for all $N \in \N$, and such that $F_N^{(m)}=\log I_N^{(m)}$ satisfies
		
			\begin{equation*}
				\lim_{N \to \infty} N^{2k-2} \left\|
				F_N^{(m)} - \sum_{g=0}^k N^{2-2g} F_{Ng}^{(m)} \right\|_\varepsilon = 0
			\end{equation*}
			
		\noindent
		for each fixed $k \in \N_0$.
	\end{thm}
	
		\subsection{Large $N$ expansion}
		For each $N \in \N$, let $\rho_N>0$ be such that the entire function $I_N$ is non-vanishing 
		on the closed origin-centered polydisc $\D_N(\rho_N)$. 
		Note that we are not assuming the positive sequence $(\rho_N)_{N=1}^\infty$
		can be selected such that it has positive infimum, i.e. we are not assuming that the stable
		non-vanishing hypothesis holds.
		
		For each $N \in \N$, the free energy $F_N=\log I_N$ is analytic on an open neighborhood of $\D_N(\rho_N)$, 
		and its Maclaurin series 
						
			\begin{equation}
				F_N^d = \sum_{d=1}^\infty \frac{z^d}{d!} F_N^d,
			\end{equation}
			
		\noindent
		converges uniformly absolutely on $\D_N(\rho_N)$.
				
			\begin{thm}
			\label{thm:StrongCouplingConvergent}
				For any $1 \leq d \leq N$, we have 
				
					\begin{equation*}
						F_N^d = \sum_{g=0}^\infty N^{2-2g} F_{Ng}^d 
					\end{equation*}
					
				\noindent
				where the series converges $\|\cdot\|$-absolutely and
				$F_{Ng}^d$ is the polynomial \eqref{eqn:GenusSpecificPolynomials}.			
			\end{thm}
			
			\begin{proof}
				This follows immediately from Theorem \ref{thm:TopologicalExpansion}
				together with Theorem \ref{thm:ExponentialFormula}.
			\end{proof}
			
		Theorem \ref{thm:StrongCouplingConvergent} yields the large $N$ expansion 
		(aka genus expansion or `t Hooft expansion)
		of each fixed strong coupling coefficient $F_N^d$, as defined
		by \eqref{eqn:LargeNExpansionPrecise}.

			\begin{thm}
			\label{thm:LargeNExpansion}
				For each fixed $d \in \N$ and $k \in \N_0$, we have 
				
					\begin{equation*}
						\lim_{N \to \infty} N^{2k-2} \left\| F_N^d - \sum_{g=0}^k N^{2-2g} F_{Ng}^d \right\| = 0
					\end{equation*}
			\end{thm}
			
			\begin{proof}
				Let $d \in \N$ and $k \in \N_0$ be fixed. By Theorem \ref{thm:StrongCouplingConvergent},
				for any $N \geq d$ we have
				
					\begin{equation*}
						F_N^d - \sum_{g=0}^k N^{2-2g} F_{Ng}^d = \sum_{g=k+1}^\infty N^{2-2g} F_{Ng}^d.
					\end{equation*}
					
				\noindent
				Moreover, by Theorem \ref{thm:ConnectedSorting} we have 
				
					\begin{equation}
						\|F_{Ng}^d\| < 4^d \mon_g^d.
					\end{equation}
					
				\noindent
				We thus have
				
					\begin{equation*}
						N^{2k-2}\left\| F_N^d - \sum_{g=0}^k N^{2-2g} F_{Ng}^d \right\| < 4^d 
						\sum_{l=1}^\infty N^{-2l} \mon_{k+l}^d < \infty
					\end{equation*}
					
				\noindent
				for any $N \geq d$.
				Since the upper bound is positive and strictly decreasing in $N$, the result follows from
				the monotone convergence theorem. 
			\end{proof}
			
		\subsection{Cancellation scheme}
		The fact that monotone Hurwitz numbers are the combinatorial invariants underlying the 
		large $N$ expansion of the strong coupling coefficients $F_N^d$ of the HCIZ and BGW 
		integrals has implications in both directions. In particular, it implies the following cancellation
		feature of monotone Hurwitz numbers, which classical Hurwitz numbers do not share. 
		We shall use this cancellation scheme below.
		
		\begin{thm}
			\label{thm:Cancellation}
				For any $(d,g) \in \N \times \N_0$ except $(1,0)$, 
				we have
				
					\begin{equation*}
						\sum_{\beta \vdash d} (-1)^{\ell(\beta)} \mon_g(\alpha,\beta) = 0
					\end{equation*}
					
				\noindent
				for all $\alpha \vdash d$.
			\end{thm}
			
			\begin{proof}
				The case where $d=1$ and $g>0$ is combinatorially obvious: the sum consists of 
				the single term $\mon_g(1,1)$, which vanishes as there are no walks of positive
				length in a graph with a single vertex.
				
				For $d>1$, the cancellation identity is
				obtained by turning off one of the two external fields in the HCIZ integral and appealing
				to Theorems \ref{thm:StrongCouplingConvergent} and \ref{thm:LargeNExpansion}.
				More precisely, we consider the specialization of $I_N$ in which $B$ is the identity matrix. 
				The character form of $I_N^{(1)d}$ then 
				degenerates to
				
					\begin{equation}
						I_N^d = N^d \sum_{\substack{\lambda \vdash d\\ \ell(\lambda) \leq N}}
						s_\lambda(a_1,\dots,a_N) \dim \mathsf{V}^\lambda = N^d p_{1^d}(a_1,\dots,a_N),
					\end{equation}
				
				\noindent	
				so that $I_N$ itself degenerates to the exponential function 
		
					\begin{equation}
						I_N = e^{zN p_1(a_1,\dots,a_N)},
					\end{equation}
			
				\noindent
				a fact which is also obvious from the integral representation \eqref{eqn:HCIZ}.
				The free energy $F_N=\log I_N$ in this degeneration is simply the
				polynomial $F_N=zNp_1(a_1,\dots,a_N)$, so the strong coupling 
				coefficients are 
				
					\begin{equation}
						F_N^1 = N^2 \frac{p_1(a_1,\dots,a_N)}{N}
					\end{equation}
					
				\noindent
				and $F_N^d = 0$ for $d > 1$. Thus by Theorem \ref{thm:StrongCouplingConvergent}, for any $1 < d \leq N$ we have
				that 
				
					\begin{equation}
						\sum_{\alpha \vdash d} \frac{p_\alpha(a_1,\dots,a_N)}{N^{\ell(\alpha)}} (-1)^{\ell(\alpha)}
						\sum_{\beta \vdash d} (-1)^{\ell(\beta)} \sum_{g=0}^\infty N^{2-2g} \mon_g(\alpha,\beta) = 0,
					\end{equation}
					
				\noindent
				which yields
				
					\begin{equation}
					\label{eqn:InfiniteCancellation}
						\sum_{\beta \vdash d} (-1)^{\ell(\beta)} \sum_{g=0}^\infty N^{-2g} \mon_g(\alpha,\beta) = 0
					\end{equation}
					
				\noindent
				for all $\alpha \vdash d$, by linear independence of the Newton polynomials
				
					\begin{equation}
						p_\alpha(x_1,\dots,x_N), \quad \alpha \vdash d
					\end{equation}
					
				\noindent
				in the stable range $1 \leq d \leq N$.
				
				We now proceed by induction in $g$. For $g=0$, take the $N \to \infty$ limit 
				in \eqref{eqn:InfiniteCancellation} to obtain
				
					\begin{equation}
						\sum_{\beta \vdash d} (-1)^{\ell(\beta)} \mon_0(\alpha,\beta)=0
					\end{equation}
					
				\noindent
				for all $\alpha \vdash d$. Assuming the result holds up to genus $k$, 
				\eqref{eqn:InfiniteCancellation} becomes 
				
					\begin{equation}
					\label{eqn:InfiniteCancellationInduction}
						\sum_{\beta \vdash d} (-1)^{\ell(\beta)} \sum_{g=k+1}^\infty N^{-2g} \mon_g(\alpha,\beta) = 0,
					\end{equation}
					
				\noindent
				for all $\alpha \vdash d$.
				Multiply \eqref{eqn:InfiniteCancellationInduction} by $N^{2k}$ and take the $N \to \infty$ limit to 
				obtain
				
					\begin{equation}
						\sum_{\beta \vdash d} (-1)^{\ell(\beta)} \mon_{k+1}(\alpha,\beta)=0
					\end{equation}
					
				\noindent
				for all $\alpha \vdash d$.
					
			\end{proof}
			
			\begin{remark} 
			When $\alpha=d$ is the Young diagram consisting of a single row of $d$
			cells, we have the product formula
			
				\begin{equation}
					\mon_0(d,\beta) = \prod_{i=1}^{\ell(\beta)} \Cat_{\beta_i-1},
				\end{equation}		
				
			\noindent
			where $\Cat_k=\frac{1}{k+1}{2k \choose k}$ is the Catalan number \cite{MN2}. Thus in the case $\alpha=d$
			and $g=0$, Theorem \ref{thm:Main} becomes
			
				\begin{equation}
					\sum_{\beta \vdash d} (-1)^{\ell(\beta)} \Cat_{\beta_i-1} = 0,
				\end{equation}
				
			\noindent
			which is just the vanishing identity for the summation of the Mobius function of a poset 
			in the case of the lattice of noncrossing partitions of $\{1,\dots,d\}$; see e.g. \cite{NS}. Probably, Theorem \ref{thm:Cancellation}
			is indicative of a relationship between monotone Hurwitz numbers and Mobius 
			functions of higher genus noncrossing partitions.
			\end{remark}

	\subsection{Reduction to uniform boundedness} 
	With Theorem \ref{thm:LargeNExpansion} in hand, the proof of
	Theorem \ref{thm:MainGeneralized} reduces
	to establishing stable nonvanishing together with uniform 
	boundedness. More precisely, we have the following reduction.
				
		\begin{thm}
		\label{thm:ReductionToUniformBoundedness}
		Suppose there exists $\delta \in (0,\gamma]$ such that $I_N$ is nonvanishing on $\D_N(\delta)$ for all $N \in \N$, 
		and set $F_N=\log I_N \in \O_N(\delta)$. If for each $k \in \N_0$ we have
					
				\begin{equation*}
				\label{eqn:MainBounded}
					\left\| F_N-\sum_{g=0}^k N^{2-2g} F_{Ng} \right\|_\delta \leq M_k N^{2-2k}
				\end{equation*}
				
		\noindent
		for all $N \geq N_k$ sufficiently large, where $M_k \geq 0$ depends only on $k$, 
		then Theorem \ref{thm:MainGeneralized} holds.
		\end{thm}
					
		\begin{proof}		
		Since $F_{Ng} \in \O_N(\gamma)$ for all $g \in \N_0$ and $\delta \leq \gamma$
		the differences
		
			\begin{equation}
				\Delta_{Nk} = F_N - \sum_{g=0}^k N^{2-2g} F_{Ng}, \quad k \in \N_0,
			\end{equation}
			
		\noindent
		belong to the Banach algebra $\O_N(\delta)$. Let
		
			\begin{equation}
				\Delta_{Nk} = \sum_{d=1}^\infty \frac{z^d}{d!} \Delta_{Nk}^d
			\end{equation}
			
		\noindent
		be the coupling expansion of this holomorphic difference, so that 
		
			\begin{equation}
				\Delta_{Nk}^d = F_N^d - \sum_{g=0}^k N^{2-2g} F_{Ng}^d,
			\end{equation}
		
		\noindent
		and Theorem \ref{thm:LargeNExpansion} says that 
		
			\begin{equation}
				\lim_{N \to \infty} N^{2k-2}\|\Delta_{Nk}^d\|=0.
			\end{equation}

		Now fix $k \in \N_0$ and $\varepsilon < \delta,$ and let $\kappa>0$ be given. 
		Under the hypothesis that $\|\Delta_{Nk}\|_\delta \leq M_kN^{2-2k}$ for all $N\geq N_k$ sufficiently large, where 
		$M_k \geq 0$ depends only on $k$, we will prove that in fact
				
			\begin{equation}
				\|\Delta_{Nk}\|_\varepsilon \leq \kappa N^{2-2k}
			\end{equation}
			
		\noindent
		for all $N$ sufficiently large.
		
		For any $N \in \N$, we have that
			
				\begin{equation}
					\|\Delta_{Nk}\|_\varepsilon \leq \sum_{d=1}^n \frac{\varepsilon^d}{d!} \| \Delta_{Nk}^d \| 
					+ \sum_{d=n+1}^\infty \frac{\varepsilon^d}{d!} \| \Delta_{Nk}^d \| 
				\end{equation}
				
		\noindent
		for all $n \in \mathbb{N}$. By Cauchy's estimate, for each $d \in \mathbb{N}$ we have
			
			\begin{equation}
				\frac{\|\Delta_{Nk}^d\|}{d!} \leq \frac{\|\Delta_{Nk}\|_\delta}{\delta^d},
			\end{equation}
			
		\noindent
		and thus
			
			\begin{equation}
				\left\| \Delta_{Nk} \right\|_{\varepsilon} \\
				\leq \  \sum_{d=1}^n \frac{\varepsilon^d}{d!}\|\Delta_{Nk}^d\|
				 + \left( \frac{\varepsilon}{\delta} \right)^{n+1} \frac{M_k}{1-\frac{\varepsilon}{\delta}} N^{2-2k},
			\end{equation}

		\noindent
		holds for each $n \in \N$, for all $N \geq N_k$. Consequently, for all $N \geq N_k$ sufficiently large we have
		that
			
			\begin{equation}
				\left\| \Delta_{Nk} \right\|_{\varepsilon} \leq \  \sum_{d=1}^{n_0} \frac{\varepsilon^d}{d!}\|\Delta_{Nk}^d\| + \frac{\kappa}{2}N^{2-2k},
			\end{equation}
			
		\noindent
		where $n_0$ is sufficiently large so that 
		
			\begin{equation}
				 \left( \frac{\varepsilon}{\delta} \right)^{n_0+1} \frac{M_k}{1-\frac{\varepsilon}{\delta}} \leq \frac{\kappa}{2}.
			\end{equation}
			
		\noindent
		 Invoking Theorem \ref{thm:LargeNExpansion}, for $N$ sufficiently large we have

				\begin{equation}
					 \sum_{d=1}^{n_0} \frac{\varepsilon^d}{d!}\|\Delta_{Nk}^d\| = N^{2-2k} \sum_{d=1}^{n_0} \frac{\varepsilon^d}{d!}N^{2k-2}\|\Delta_{Nk}^d\| 
					 \leq  N^{2-2k} \frac{\kappa}{2} 
				\end{equation}
				
		\noindent
		and we conclude that 
			
				\begin{equation}
					\left\| \Delta_N \right\|_{\varepsilon}  \leq \kappa N^{2-2k},
				\end{equation}
				
		\noindent
		holds for $N$ sufficiently large, as required.

	\end{proof}		
					
	\subsection{Reduction to concentration}
	We now show that the proof of Theorem \ref{thm:MainGeneralized} can be further 
	reduced to establishing a large $N$ version of the $N=\infty$ topological concentration 
	inequalities \eqref{eqn:StableTopologicalConcentration} for the stable integrals $I$.
	
	For each $k \in \N_0$, we define $E_{N\overline{k}} \in \O_N(\gamma)$ by
		
			\begin{equation}
				E_{N\overline{k}} = e^{\sum_{g=0}^k N^{2-2g} F_{Ng}}.
			\end{equation}
			
		\noindent
		Then 
		
			\begin{equation}
				E_{N\overline{k}} = 1 + \sum_{d=1}^\infty \frac{z^d}{d!} E_{N\overline{k}}^d,
			\end{equation}
			
		\noindent
		where 
				
			\begin{equation}
			\begin{split}
				E_{N\overline{k}}^{(0)d} &= \sum_{g=-d+1}^\infty N^{2-2g}\mon_{g\overline{k}}^{\bullet d}, \\
				E_{N\overline{k}}^{(1)d} &= \sum_{g=-d+1}^\infty N^{2-2g}  \sum_{\alpha \vdash d} \frac{p_\alpha(a_1,\dots,a_N)}{N^{\ell(\alpha)}}
				 (-1)^{\ell(\alpha)+d}\mon_{g\overline{k}}^\bullet(\alpha), \\
				E_{N\overline{k}}^{(2)d} &= \sum_{g=-d+1}^\infty N^{2-2g}
				\sum_{\alpha,\beta \vdash d} \frac{p_\alpha(a_1,\dots,a_N)}{N^{\ell(\alpha)}}\frac{p_\beta(b_1,\dots,b_N)}{N^{\ell(\beta)}}
				(-1)^{\ell(\alpha)+\ell(\beta)}\mon_{g\overline{k}}^\bullet(\alpha,\beta),
			\end{split}
			\end{equation}
			
		\noindent
		is a genus expansion for degree $d$ disconnected covers built from connected
		components of genus at most $k$ which
		converges uniformly absolutely on compact subsets of $\C^{mN}$, for all $d \in \N$.
		In particular, we have
		
			\begin{equation}
				\|E_{N\overline{k}}^d\| = N^{2d}\left(1+O_k\left(\frac{1}{N}\right)\right).
			\end{equation}
			
		\noindent
		Note that, in the stable range $1 \leq d \leq N$, the complementary genus expansions 
		
			\begin{equation}
			\begin{split}
				E_{N\underline{k+1}}^{(0)d} &= \sum_{g=k+1}^\infty N^{2-2g}\mon_{g\underline{k+1}}^{\bullet d}, \\
				E_{N\underline{k+1}}^{(1)d} &= \sum_{g=k+1}^\infty N^{2-2g}  \sum_{\alpha \vdash d} \frac{p_\alpha(a_1,\dots,a_N)}{N^{\ell(\alpha)}}
				 (-1)^{\ell(\alpha)+d}\mon_{g\underline{k+1}}^\bullet(\alpha), \\
				E_{N\underline{k+1}}^{(2)d} &= \sum_{g=k+1}^\infty N^{2-2g}
				\sum_{\alpha,\beta \vdash d} \frac{p_\alpha(a_1,\dots,a_N)}{N^{\ell(\alpha)}}\frac{p_\beta(b_1,\dots,b_N)}{N^{\ell(\beta)}}
				(-1)^{\ell(\alpha)+\ell(\beta)}\mon_{g\underline{k+1}}^\bullet(\alpha,\beta),
			\end{split}
			\end{equation}
			
		\noindent
		enumerating disconnnected degree $d$ covers assembled from 
		connected components of genus at least $k+1$ converge uniformly absolutely on compact subsets of $\C^{mN}$, 
		by Theorem \ref{thm:TopologicalExpansion} and the fact that 
		$\mon_{g\underline{k+1}}^\bullet(\alpha,\beta) \leq \mon_g^\bullet(\alpha,\beta)$.
		However, the generating series for disconnected covers built from connected components
		of unbounded genus are divergent in the unstable range $d > N$.
			
		For each $k \in \N_0$, define $\Phi_{Nk} \in \O_N(\gamma)$ by 
		
			\begin{equation}
				\Phi_{Nk} = E_{N\overline{k}}^{-1}I_N.
			\end{equation}
			
		\noindent
		In the case $m=1$, 
		
			\begin{equation}
				\Phi_{Nk}^{(1)} = e^{-\sum_{g=0}^k N^{2-2g} F_{Ng}^{(1)}} \int_{\mathrm{U}(N)} e^{\sqrt{z}N \mathrm{Tr}(AU + BU^{-1})} \mathrm{d}U
			\end{equation}
			
		\noindent
		is the $k$th order topological normalization of the BGW integral $I_N^{(1)}$, while in the case $m=2$
		
			\begin{equation}
				\Phi_{Nk}^{(2)} = e^{-\sum_{g=0}^k N^{2-2g} F_{Ng}^{(2)}} \int_{\mathrm{U}(N)} e^{zN \mathrm{Tr}AUBU^{-1}} \mathrm{d}U
			\end{equation}
			
		\noindent
		is the $k$th order topological normalization of the HCIZ integral $I_N^{(2)}$. 		
		
	\begin{thm}
	\label{thm:ReductionToConcentration}
		If there exists a constant $\xi \in (0,\gamma]$ such that
		for each $k \in \N_0$ we have
			
				\begin{equation*}
				    \|\Phi_{Nk}-1\|_\xi \leq C_k N^{2-2k}
				\end{equation*}
			
		\noindent
		for all $N \geq N_k$ sufficiently large, where $C_k \geq 0$ depends only on $k$,
		then Theorem \ref{thm:MainGeneralized} holds.
	\end{thm}
		
	\begin{proof}		
		By hypothesis, we have
		
			\begin{equation}
				\|\Phi_{N2}-1\|_\xi \leq C_2 N^{-2}
			\end{equation}
			
		\noindent
		for all $N \geq N_2$ sufficiently large. This implies that 
		
			\begin{equation}
				\|\Phi_{N2}-1\|_\xi <1
			\end{equation}
			
		\noindent
		for all $N > N_0:=\max (N_2,\sqrt{C_2})$ sufficiently large, and hence that $\Phi_{N2}$ is non-vanishing on $\D_N(\xi)$ for all $N>N_0$.
		But 
		
			\begin{equation}
				\Phi_{N2} = e^{-(N^2F_{N0} + N^0F_{N1} + N^{-2}F_{N2})} I_N
			\end{equation}
			
		\noindent
		is the product of a non-vanishing function and $I_N$, so that $I_N$ 
		must be nonvanishing on $\D_N(\xi)$ for all $N>N_0$. Consequently,
		there is $\eta \in (0,\xi]$ such that $I_N$ is nonvanishing on 
		$\D_N(\eta)$ for all $N \in \N$, whence $F_N= \log I_N$ is a well-defined 
		member of $\O_N(\eta)$ for all $N \in \N$. 
		
		Now let $k \in \N_0$ be arbitrary but fixed. Since $\eta \leq \xi \leq \gamma$,
		the difference
				
			\begin{equation}
				\Delta_{Nk} = F_N-\sum_{g=0}^k N^{2-2g} F_{Ng}
			\end{equation}
			
		\noindent
		belongs to $\O_N(\eta)$ for all $N \in \N$. Now

			\begin{equation}
				\|e^{\Delta_{Nk}} - 1\|_\eta = \|\Phi_{Nk}-1\|_\eta \leq \|\Phi_{Nk}-1\|_\xi,
		    \end{equation}
		    
		 \noindent
		 so that by hypothesis we have
		 
		 	\begin{equation}
				\|e^{\Delta_{Nk}} - 1\|_\eta \leq C_kN^{2-2k}
			\end{equation}
			
		\noindent
		 for all $N \geq N_k$ sufficiently large. 
		 
		For the rest of the argument, we assume $N \geq N_k$. From the above we have that $e^{\Delta_{Nk}}$ belongs to the 
		closed unit ball of radius $C_kN^{2-2k}$ centered at the constant function $1$ in
		the Banach algebra ($\O_N(\eta),\|\cdot\|_\eta)$. By the triangle inequality, we thus have
		
			\begin{equation}
				\|e^{\Delta_{Nk}} \|_\eta \leq 1 + C_k N^{2-2k},
			\end{equation} 
			
		\noindent
		so that
		
			\begin{equation}
				e^{ R_{Nk}(\eta)} \leq 1 + C_k N^{2-2k},
			\end{equation} 

		\noindent
		where 
		
			\begin{equation}
				R_{Nk}(\eta) = \sup\limits_{\D_N(\eta)} \Re \Delta_{Nk}
			\end{equation}
			
		\noindent
		is the supremum of the real part of $\Delta_{Nk}$ over $\D_N(\eta)$.
		Thus, we have the bound 
		
			\begin{equation}
				R_{Nk}(\eta) \leq \log \left( 1 + C_k N^{2-2k} \right) \leq C_k N^{2-2k}.
			\end{equation}
			
		\noindent
		Now by the Borel-Carath\'eodory inequality \cite{Titchmarsh} 
		we have
		
			\begin{equation}
				\|\Delta_{Nk}\|_\delta \leq \frac{2\delta}{\eta-\delta} R_{Nk}(\eta)
			\end{equation}
			
		\noindent
		for any $\delta \in (0,\eta)$, and choosing $\delta = \frac{\eta}{2}$ we obtain 
		
			\begin{equation}
				\|\Delta_{Nk}\|_\delta \leq M_kN^{2-2k},
			\end{equation}
			
		\noindent
		where $M_k=2C_k$. The result thus follows from Theorem \ref{thm:ReductionToUniformBoundedness}.
	\end{proof}
	
	\subsection{Reduction to critical bounds}
	Let 
	
		\begin{equation}
			\Phi_{Nk} = 1 + \sum_{d=1}^\infty \frac{z^d}{d!} \Phi_{Nk}^d
		\end{equation}
		
	\noindent
	be the coupling expansion of the topological normalization $\Phi_{Nk} = E_{N\overline{k}}^{-1}I_N \in \O_N(\gamma)$.
	Thus 
	
		\begin{equation}
			\Phi_{Nk}^d = \sum_{c=0}^d {d \choose c} I_N^c (E_{N\overline{k}}^{-1})^{d-c}
		\end{equation}
		
	\noindent
	is the binomial convolution of the coupling coefficients of the entire function 
	
		\begin{equation}
			I_N = 1 + \sum_{d=1}^{\infty} \frac{z^d}{d!} I_N^d
		\end{equation}
		
	\noindent
	and the coupling coefficients of the topological normalization factor $E_{N\overline{k}}^{-1} \in \O_N(\gamma)$,
	
		\begin{equation}
			E_{N\overline{k}}^{-1} = e^{-\sum_{g=0}^k N^{2-2g} F_{Ng}} = 1 + \sum_{d=1}^\infty \frac{z^d}{d!} (E_{N\overline{k}}^{-1})^d.
		\end{equation}
	
	\noindent
	We then have the following final reduction of Theorem \ref{thm:MainGeneralized}, which 
	reduces it to establishing norm bounds on the $\Theta(N^2)$ critical coupling coefficients
	of $\Phi_{Nk}$.
		
		\begin{thm}
		\label{thm:ReductionToCriticalBounds}
			If there exists a constant $\xi \in (0,\gamma]$ such that, 
			for each $k \in \N_0$, the inequalities 
			
				\begin{equation*}
					\frac{\xi^d}{d!}\|\Phi_{Nk}^d\| \leq R_k N^{-2k}, \quad 1 \leq d \leq tN^2,
				\end{equation*}
				
			\noindent
			hold for all $N \geq N_k$ sufficiently large, where $R_k \geq 0$
			depends only on $k$, then Theorem \ref{thm:MainGeneralized} holds.
		\end{thm} 
		
		\begin{proof}
		First, we have the topologically normalized version of Theorem \ref{thm:StrongCouplingApproximation}:
		for sufficiently strong coupling, the terms of degree $d > tN^2$ in the coupling expansion
		
			\begin{equation}
				\Phi_{Nk} = 1 +\sum_{d=1}^\infty \frac{z^d}{d!} \Phi_{Nk}^d
			\end{equation}
			
		\noindent
		can be ignored. Indeed, we have
		
			\begin{equation}
				\|\Phi_{Nk}^d\| \leq \sum_{c=0}^d {d \choose c} \|I_N^c\| \|(E_{Nk}^{-1})^{d-c}\|
				\lesssim  \text{const} \cdot N^{2d},
			\end{equation}
			
		\noindent
		so the same argument as in Theorem \ref{thm:StrongCouplingApproximation} shows that 
				
			\begin{equation}
				\left\| \sum_{d>tN^2} \frac{z^d}{d!} \Phi_{Nk}^d \right\|_\xi \leq \sum_{d>tN^2} \frac{\xi^d}{d!} \|\Phi_{Nk}^d\|
			\end{equation}
			
		\noindent
		is $O(e^{-N^2})$ for $\xi > 0$ a sufficiently small absolute constant.
		Thus, to bound $\|\Phi_{Nk}-1\|_\xi$, it is sufficient to bound
		the finite sum,
		
			\begin{equation}
				\sum_{d=1}^{\lfloor tN^2 \rfloor} \frac{z^d}{d!} \Phi_{Nk}^d,
			\end{equation}
			
		\noindent
		and by hypothesis we have 
		
			\begin{equation}
				\left\| \sum_{d=1}^{\lfloor tN^2 \rfloor} \frac{z^d}{d!} \Phi_{Nk}^d \right\|_\xi
				\leq \sum_{d=1}^{\lfloor tN^2 \rfloor} \frac{\xi^d}{d!} \|\Phi_{Nk}^d\| 
				\leq \sum_{d=1}^{\lfloor tN^2 \rfloor} R_k N^{-2k} \leq tR_k N^{2-2k}.
			\end{equation}

		\end{proof}

	\subsection{Stable critical bounds}
	We now show that the required uniform bounds 
	on the critical coupling coefficients $\Phi_{Nk}^d$ hold
	in the stable range, $1 \leq d \leq N$.
	
		\begin{thm}
		\label{thm:StableCriticalBounds}
			For any integers $1 \leq d \leq N$ and $k \in \N_0$, we have
			
				\begin{equation*}
					\|\Phi_{Nk}^d\| \leq R_k N^{-2k},
				\end{equation*}
				
			\noindent
			where $R_k \geq 0$ depends only on $k$.
		\end{thm}
	
		\begin{proof}
		Let the integers $1 \leq d \leq N$ and $k \in \N_0$ be arbitrary but fixed,
		and consider
		
			\begin{equation}
				\Phi_{Nk}^d = \sum_{c=0}^d {d \choose c} I_N^c (E_{Nk}^{-1})^{d-c}.
			\end{equation}
		
		\noindent	
		By Theorem \ref{thm:TopologicalExpansion}
		together with Theorem \ref{thm:StableTopologicalConvolution},
		we have 
		
			\begin{equation}
				I_N^c = \sum_{b=0}^c {c\choose b} E_{N\overline{k}}^b E_{N\underline{k+1}}^{c-b}, \quad 1 \leq c \leq d,
			\end{equation}
			
		\noindent
		so that
		
			\begin{equation}
				\Phi_{Nk}^d = E_{N\overline{k+1}}^d,
			\end{equation}
			
		\noindent
		where $E_{N\underline{k+1}}^d=E_{N\underline{k+1}}^{(m)d}$ is a $\|\cdot\|$-absolutely convergent
		generating function for disconnected degree $d$ covers of genus at 
		least $k+1$ with at most $m$ non-simple branch points, $m\in \{0,1,2\}$. 
		Thus by Theorem \ref{thm:DisconnectedSorting}, we have
			
			\begin{equation}
				\|\Phi_{Nk}^d\| < 4^d \sum_{g=k+1}^\infty N^{2-2g}\mon_g^{\bullet d},
			\end{equation}
		
		\noindent
		so that the inequality 
	
			\begin{equation}
				\frac{\xi^d}{d!}\|\Phi_{Nk}^{(m)d}\| \leq \frac{(4\xi)^d}{d!} \sum_{g=k+1}^\infty N^{2-2g} \mon_g^{\bullet d}
			\end{equation}
		
		\noindent
		holds for any $\xi \geq 0$. 
		
		It remains to prove that for $\xi >0$ a sufficiently small
		absolute constant, we have 
		
		\begin{equation}
		\label{eqn:Penultimate}
			\frac{(4\xi)^d}{d!} \sum_{g=k+1}^\infty N^{2-2g} \mon_g^{\bullet d} \leq R_k N^{-2k}
		\end{equation}
		
		\noindent
		where $R_k \geq 0$ depends only on $k$.
		From Theorem \ref{thm:LargeGenusAsymptotics}, we get that
	
		\begin{equation}
		\begin{split}
			\frac{(4\xi)^d}{d!} \mon_g^{\bullet d} &\leq \frac{(4\xi)^d(d-1)^{3d-3}}{(d-1)!d!d!} (d-1)^{2g} \\
			&= \frac{(4\xi)^d(d-1)^{3d-3}}{(d-1)!(d-1)!(d-1)!} \frac{(d-1)^{2g}}{d^2} \\
			&< (4\xi)^d e^{3d-3} (d-1)^{2g-2} \\
			& < e^{-d} (d-1)^{2g-2} 
		\end{split}
		\end{equation}
		
	\noindent
	for $\xi >0$ a sufficiently small absolute constant. 
	We thus have the bound 
	
		\begin{equation}
			 \frac{(4\xi)^d}{d!}\sum_{g=k+1}^\infty N^{2-2g} \mon_g^{\bullet d} \leq 
			 e^{-d} \sum_{g=k+1}^\infty \left( \frac{d-1}{N} \right)^{2g-2},
		\end{equation}
		
	\noindent
	where the geometric series converges because $d \leq N$. 
	Factoring out the first term of the series and summing, we get
	
		\begin{equation}
		\begin{split}
			 e^{-d} \sum_{g=k+1}^\infty \left( \frac{d-1}{N} \right)^{2g-2} &= e^{-d}(d-1)^{2k} N^{-2k}  \sum_{l=0}^\infty \left( \frac{d-1}{d} \right)^{2l} \\
			& = e^{-d} (d-1)^{2k} d N^{-2k} \\
			&< e^{-d} d^{2k+1} N^{-2k}.
		\end{split}
		\end{equation}
		
	\noindent
	Since 
	
		\begin{equation}
			 e^{-x} x^{2k+1} \leq \left( \frac{2k+1}{e} \right)^{2k+1}, \quad x \in \R,
		\end{equation}
		
	\noindent
	we conclude that 
	
		\begin{equation}
			\frac{\xi^d}{d!}\|\Phi_{Nk}^d\| \leq R_k N^{-2k}
		\end{equation}
		
	\noindent
	holds with
	
		\begin{equation}
			R_k = \left( \frac{2k+1}{e} \right)^{2k+1}.
		\end{equation}
	\end{proof}
		
	\subsection{The Plancherel mechanism}	
	It remains to extend Theorem \ref{thm:StableCriticalBounds} into 
	the unstable critical range, $N < d \leq tN^2$. The reason that the above argument
	cannot be repeated verbatim is that the convolution

		\begin{equation}
			\Phi_{Nk}^d = \sum_{c=0}^d {d \choose c} I_N^c (E_{Nk}^{-1})^{d-c},
		\end{equation}
				
	\noindent
	contains unstable terms as soon as $d>N$, i.e. terms containing the factor
	$I_N^c$, $c>N$, to which the topological expansion of Theorem \ref{thm:TopologicalExpansion}
	does not apply. However, this obstruction can be overcome in the 
	unstable critical range, where $d$ is not arbitrarily large relative to 
	$N$ but subject to $d \leq tN^2$. The basic mechanism is easy to 
	describe: the Plancherel measure on Young diagrams with $d$ cells
	concentrates \cite{BDJ,DZ,Johansson:MRL,Kerov} on diagrams of height and width at most 
	$2\sqrt{d}$, which for $d < \frac{1}{4}N^2$ implies
	containment in the square diagram with $N$ rows and $N$ columns.
	The content of every cell in any Young diagram $\lambda$ contained in the $N \times N$
	square is at most $N-1$ in absolute value, and consequently the $1/N$ expansion
	
		\begin{equation}
			\Omega_{\frac{1}{N}}^{-1}(\lambda) =\prod_{\Box \in \lambda} \frac{1}{1+\frac{c(\Box)}{N}}
			= \sum_{r=0}^\infty \left( - \frac{1}{N} \right)^r f_r(\lambda)
		\end{equation}
		
	\noindent
	is absolutely convergent. Summarizing this in the notationally
	simplest case $m=0$, if $\lambda \vdash d < \frac{1}{4}N^2$ then most terms
	in the sum 
	
		\begin{equation}
			I_N^d = \sum_{\substack{\lambda \vdash d \\ \ell(\lambda) \leq N}} 
			\Omega_{\frac{1}{N}}^{-1}(\lambda) \frac{(\dim \mathsf{V}^\lambda)^2}{d!}
		\end{equation}
		
	\noindent
	admit an absolutely convergent $1/N$ expansion, and those which do not are suppressed 
	by the Plancherel weight, so that the argument used in Theorem \ref{thm:StableCriticalBounds}
	applies up to negligible terms.
	
	Let us illustrate the above mechanism in explicit detail for the first unstable normalized coupling coefficient, corresponding to $d=N+1$,
	remaining in the notationally simplest case $m=0$.
	We have
		
		\begin{equation}
			\frac{\xi^{N+1}}{(N+1)!}\Phi_{Nk}^{N+1} = 
			\frac{\xi^{N+1}}{(N+1)!}\sum_{c=0}^N {N+1 \choose c} I_N^c(E_{Nk}^{-1})^{N+1-c}
			+\frac{\xi^{N+1}}{(N+1)!}I_N^{N+1},
		\end{equation}
		
	\noindent
	where the final term corresponding to $c=N+1$ involves the unstable coupling coefficient
	
		\begin{equation}
			I_N^{N+1} = N^{2N+2}\sum_{\substack{\lambda \vdash N+1 \\ \ell(\lambda) \leq N}} 
			\Omega_{\frac{1}{N}}^{-1}(\lambda) \frac{(\dim \mathsf{V}^\lambda)^2}{(N+1)!}.
		\end{equation}
		
	\noindent
	The term in this sum corresponding to
	the row diagram $\lambda=(N+1)$ contains the factor
	
		\begin{equation}
			\Omega_{\frac{1}{N}}^{-1}(N+1) = \frac{1}{(1+\frac{1}{N})\dots(1+\frac{N}{N})},
		\end{equation}
		
	\noindent
	which does not admit an absolutely convergent $1/N$ expansion due to the oscillatory 
	divergence
	
		\begin{equation}
			 \frac{1}{1+\frac{N}{N}} = 1-1+1-1+\dots.
		\end{equation}

	\noindent
	However, the net contribution of this $\lambda = (N+1)$ term,
	
		\begin{equation}
			\frac{\xi^{N+1}}{(N+1)!} \frac{N^{2N+2}}{(1+\frac{1}{N})\dots(1+\frac{N}{N})} \frac{1}{(N+1)!},
		\end{equation}
		
	\noindent
	is made small both by the coupling prefactor $\frac{\xi^{N+1}}{(N+1)!}$ and the Plancherel weight 
	$\frac{1}{(N+1)!}$, and in particular is exponentially small in $N$ for $\xi>0$ a small enough constant.
	
	We thus have
	
		\begin{equation}
			\frac{\xi^{N+1}}{(N+1)!}\Phi_{Nk}^{N+1} = \frac{\xi^{N+1}}{(N+1)!} \sum_{c=0}^{N+1}
			{N+1 \choose c} \left( N^{2c}\sum_{r=0}^\infty \left(-\frac{1}{N} \right)^r \sum_{\substack{\lambda \vdash c \\ \lambda \neq (N+1)}} f_r(\lambda) 
			\frac{(\dim \mathsf{V}^\lambda)^2}{c!} \right) (E_{Nk}^{-1})^{N+1-c} + O(e^{-N}),
		\end{equation}
		
	\noindent
	with the infinite series at the $c$th term of the convolution being absolutely convergent. 
	Furthermore, for every term except $c=N+1$, the coefficients of the $1/N$ expansion are the 
	Plancherel averages $\langle f_r \rangle$, which are disconnected monotone simple Hurwitz
	numbers.  The $c=N+1$ term is the absolutely convergent series 
	
		\begin{equation}
		\label{eqn:TruncatedExpectations}
			\frac{\xi^{N+1}}{(N+1)!} \sum_{r=0}^\infty \left(-\frac{1}{N} \right)^r \sum_{\substack{\lambda \vdash N+1 \\ \lambda \neq (N+1)}} f_r(\lambda) 
			\frac{(\dim \mathsf{V}^\lambda)^2}{(N+1)!}, 
		\end{equation}
		
	\noindent
	whose coefficients are not Plancherel expectations, and therefore not monotone Hurwitz numbers. 
	The sum will of course become divergent again if we attempt to complete the sum over Young diagrams
	to a Plancherel expectation at each order of the $1/N$ expansion. However, we can complete any finite number 
	terms in the $1/N$ expansion without losing convergence --- in particular, we can complete
	the coefficients of the $1/N$ expansion out to order $r=2k-2+2d$ to match the Riemann-Hurwitz formula
	and get a topological expansion to genus $k$. The question is how much this completion costs as a function
	of $N$ and $k$.
	
	More precisely, for each $r \in \N_0$ we have
	
		\begin{equation}
			\sum_{\substack{\lambda \vdash N+1 \\ \lambda \neq (N+1)}} f_r(\lambda) 
			\frac{(\dim \mathsf{V}^\lambda)^2}{(N+1)!} = \sum_{\lambda \vdash N+1} f_r(\lambda) 
			\frac{(\dim \mathsf{V}^\lambda)^2}{(N+1)!} - f_r(1,2,\dots,N)\frac{1}{(N+1)!},
		\end{equation}
		
	\noindent
	so that completion of the sum over Young diagrams at order $1/N^r$ to the Plancherel expectation 
	$\langle f_r \rangle$ costs
	
		\begin{equation}
			f_r(1,2,\dots,N)\frac{1}{(N+1)!} = \stirling{N+r}{N} \frac{1}{(N+1)!} < N^r {N + r \choose N} \frac{1}{(N+1)!},
		\end{equation}
		
	\noindent
	where we have used the standard bound on Stirling numbers of the second kind by the 
	corresponding binomial coefficient. Thus, completing the truncated Plancherel expectations
	in \eqref{eqn:TruncatedExpectations} to order $r=2k-2+2d = 2k-2+2(N+1)=2k+2N$ has cost bounded by 
	
		\begin{equation}
			\frac{\xi^{N+1}}{(N+1)!(N+1)!} \sum_{r=0}^{2k+2N} \left(-\frac{1}{N}\right)^r N^r 
			{N+r \choose N},
		\end{equation}
		
	\noindent
	which is in turn bounded by 
	
		\begin{equation}
			\frac{\xi^{N+1}}{(N+1)!(N+1)!} \sum_{r=0}^{2k+2N} {N+r \choose N} = \frac{\xi^{N+1}}{(N+1)!(N+1)!} {3N+2k+1 \choose N+1},
		\end{equation}
		
	\noindent
	where the sum has been evaluated using the hockey stick identity, 
	and the overall result is again negligible due to the small coupling factor and the 
	Plancherel weight. Thus up to negligible terms the topologically normalized coupling 
	coefficient $\frac{\xi^{N+1}}{(N+1)!}\Phi_{Nk}^{N+1}$ is given by
	
		\begin{equation}
			\frac{\xi^{N+1}}{(N+1)!}\sum_{c=0}^{N+1} {N+1 \choose c} 
			\left(\sum_{g=-c+1}^k N^{2-2g} \mon_g^{\bullet c} \right) (E_{N\overline{k}}^{-1})^{N+1-c},
		\end{equation}
		
	\noindent
	which cancels up to a tail sum which is $O_k(N^{-2k})$ by the same argument
	used in Theorem \ref{thm:StableCriticalBounds}.
		
	A further illustration is provided by the case of trivial external fields for $I_N^{(1)}$ and $I_N^{(2)}$.
	Recall from the proof of Theorem \ref{thm:StrongCouplingApproximation} that
	the evaluation of the BGW coupling coefficients $I_N^{(1)d}$ at $(1,\dots,1)$ 
	is given by the sums
		
		\begin{equation}
			L_N^d = N^{2d} \sum_{\substack{\lambda \vdash d \\ 
			\ell(\lambda) \leq N}} \frac{(\dim \mathsf{V}^\lambda)^2}{d!},
		\end{equation}
		
	\noindent
	while the evaluation of the HCIZ coupling coefficients $I_N^{(2)d}$ 
	at $(1,\dots,1)$ is given by 
	
		\begin{equation}
			E_N^d = N^{2d}.
		\end{equation}
	
	\noindent	
	In order to proceed, we also need to evaluate the $\|\cdot\|$-absolutely 
	convergent series 
	
		\begin{equation}
		\begin{split}
			E_{N\overline{k}}^{(1)d} &= \sum_{g=-d+1}^\infty N^{2-2g} \sum_{\alpha \vdash d} (-1)^{\ell(\alpha)+d}
			\frac{p_\alpha(a_1,\dots,a_N)}{N^{\ell(\alpha)}}\mon_{g\overline{k}}(\alpha) \\
			E_{N\overline{k}}^{(2)d} &= \sum_{g=-d+1}^\infty N^{2-2g} \sum_{\alpha,\beta \vdash d} (-1)^{\ell(\alpha)+\ell(\beta)}
			\frac{p_\alpha(a_1,\dots,a_N)}{N^{\ell(\alpha)}} \frac{p_\beta(b_1,\dots,b_N)}{N^{\ell(\beta)}}\mon_{g\overline{k}}(\alpha,\beta) 
		\end{split}
		\end{equation}
	
	\noindent
	at the point $(1,\dots,1)$. Observe that by Theorem \ref{thm:Cancellation}, we have
	
		\begin{equation}
			F_{Ng}^{(1)}(z,1,\dots,1) = F_{Ng}^{(2)}(z,1,\dots,1) = \delta_{g0} z,
		\end{equation}
		
	\noindent
	so that 
	
		\begin{equation}
			E_{N\overline{k}}^{(1)}(z,1,\dots,1) = E_{N\overline{k}}^{(2)}(z,1\dots,1) = e^{N^2z}
		\end{equation}
		
	\noindent
	for all $k \in \N_0$. Thus, the topologically normalized coupling coefficients $\Phi_{Nk}^{(1)d}$ and
	$\Phi_{Nk}^{(2)d}$ evaluated at $(1,\dots,1)$ are simply the binomial convolutions 
	
		\begin{equation}
			\Phi_{Nk}^{(1)d}(1,\dots,1) = \sum_{c=0}^d {d\choose c} L_N^c (E_N^{-1})^{d-c} 
		\end{equation}
	
	\noindent
	and 
	
		\begin{equation}
			\Phi_{Nk}^{(2)d}(1,\dots,1) = \sum_{c=0}^d {d\choose c} E_N^c (E_N^{-1})^{d-c},
		\end{equation}
	
	\noindent
	for all $k \in \N_0$. In the $m=2$ case, we have perfect cancellation for all values of $d$,
	
		\begin{equation}
			\Phi_{Nk}^{(2)d}(1,\dots,1) =(N^2-N^2)^d = 0,
		\end{equation}
		
	\noindent
	but in the case $m=1$ this only holds in the stable range $1 \leq d \leq N$, where $L_N^d=E_N^d=N^{2d}$.
	For general $d$ we have
	
		\begin{equation}
			\frac{\xi^d}{d!}\Phi_{Nk}^{(1)d}(1,\dots,1) = \frac{\xi^d}{d!} \sum_{c=0}^d {d\choose c} \left(N^{2c} 
			\sum_{\substack{\lambda \vdash c \\ \ell(\lambda) \leq N}} \frac{(\dim \mathsf{V}^\lambda)^2}{c!} \right)(-N^2)^{d-c},
		\end{equation}
		
	\noindent
	and the fact that this is $O_k(N^{-2k})$ for any $k \in \N_0$ provided $d < \frac{1}{4}N^2$ is 
	precisely Lemma 2.3 in \cite{Johansson:MRL}. In particular, the fieldless case of 
	Theorem \ref{thm:Main} for the BGW integral $I_N^{(1)}$ is the following.
				
		\begin{thm}
		\label{thm:GrossWitten}
			There exists $\varepsilon > 0$ such that $L_N$ is 
			nonvanishing on $\D(\varepsilon) \subset \C$ for all
			$N \in \N$, and for each $k \in \N_0$ we have
				
				\begin{equation*}
					\lim_{N \to \infty} N^{2k-2}\left\| \log L_N - N^2z \right\|_\varepsilon =0.
				\end{equation*}
		\end{thm}
	
	We remark that the lack of higher-order corrections
	to the large $N$ limit of the fieldless BGW free energy has been the source of some 
	amazement \cite{Goldschmidt,Samuel}, but as we have explained here this is a direct
	consequence of the fact that the fieldless HCIZ integral is simply an exponential function, which 
	together with Theorem \ref{thm:LargeNExpansion} implies Theorem \ref{thm:Cancellation}.
	This underscores how clarifying the unified description of these integrals 
	as generating functions for monotone single and double Hurwitz numbers presented here can be.

	\bibliographystyle{amsplain}

\end{document}